\documentclass[12pt, reqno]{amsart}
\usepackage{graphicx}
\usepackage{subcaption}
\usepackage{amsfonts,amsmath,amsthm,amssymb,color}
\usepackage{mathtools}
\usepackage{fullpage}
\usepackage{hyperref}
\usepackage{bbm}
\usepackage{float}
\usepackage{subcaption}

\usepackage[backend=biber,sorting=nty,style=numeric,maxbibnames=99]{biblatex}
\newcommand{\N}{\mathbb{N}}

\newtheorem{theorem}{Theorem}[section]

\newtheorem{lemma}[theorem]{Lemma}
\newtheorem{corollary}[theorem]{Corollary}
\newtheorem{remark}[theorem]{Remark}

\newtheorem{example}[theorem]{Example}
\newtheorem{fact}[theorem]{Fact}

\title{
Who's the GOAT? Sports Rankings and Data-Driven Random Walks on the Symmetric Group
}
\date{\today}

\author[Garcia]{Gian-Gabriel P. Garcia}
\address[Gian-Gabriel P. Garcia]{H. Milton Stewart School of Industrial and Systems Engineering, Georgia Institute of Technology}
\email{\textcolor{blue}{\href{mailto:giangarcia@gatech.edu}{giangarcia@gatech.edu}}}

\author[Mart\'inez Mori]{J. Carlos Mart\'inez Mori}
\address[J.~C. Mart\'inez Mori]{H. Milton Stewart School of Industrial and Systems Engineering, Georgia Institute of Technology}
\email{\textcolor{blue}{\href{mailto:jcmm@gatech.edu}{jcmm@gatech.edu}}}

\begin{document}


\begin{abstract}
Given a collection of historical sports rankings, can one tell which player is the greatest of all time (i.e., the GOAT)?
In this work, we design a data-driven random walk on the symmetric group to obtain a stationary distribution over player rankings, spanning across different time periods in sports history.
We combine this distribution with a notion of stochastic dominance to obtain a partial order over the players.
Compared to existing methods, our approach is distinct in that \textit{i}) using historical rankings captures the evolution of value systems and facilitates player comparisons when head-to-head data is unavailable, and \textit{ii}) aggregating into a partial order formally comes to terms with the possibility that, while some player comparisons can be established conclusively, others are ``too close to call.''
We implement our methods using publicly available data from the Association of Tennis Professionals (ATP) and the Women's Tennis Association (WTA).
Our main findings indicate that Steffi Graf and Serena Williams are the ones that come ahead as the GOATs of the WTA.
Likewise,  the ``Big Three,'' that is Novak Djokovic, Roger Federer, and Rafael Nadal, are the ones that come ahead as the GOATs of the ATP. 
As a secondary finding, we note major differences in terms of career and dominance longevity for top players across the associations.
While initially motivated by this application in sports analytics, our methods can also be applied to other practical domains where deriving rankings from historical data can inform operational decisions, such as route planning logistics.
\end{abstract}

\maketitle
    
\section{Introduction}
\label{sec: introduction}

The question of which player is the \emph{greatest of all time} (i.e., the GOAT) is one that brings about strong opinions from sports fanatics. 
Within individual sports in particular, a wide variety of individual performance criteria can be used to argue in favor or against a given player being the GOAT.
These criteria are often quantitative at face value, taking the form of individual performance statistics. 
However, these statistics need not be consistent with one another, and thus they ultimately prompt subjective value judgments on their relative importance.
To add to this, other qualitative aspects about how a specific sport develops over time might alter their significance across generations.

For instance, consider Steffi Graf and Serena Williams; two widely successful players in women's professional tennis.
Between the two, Williams leads in the number of \emph{major} titles, holding 23 compared to the 22 held by Graf.
However, Graf leads in the overall number of titles, holding 107 compared to the 73 held by Williams.
Which of these statistics\textemdash the number of titles or the number of major titles\textemdash is more important?
How does $1$ additional major title fare against $34$ additional titles?
Note that these players had little overlap in competition: Graf played professionally from 1982 to 1999, whereas Williams turned professional in 1995 and only retired recently in 2022.
How do the numerous differences in professional tennis from the 1980s to the 2020s (e.g., changes in equipment, training, court surfaces, and point systems; consolidation of prestige around certain tournaments; increased competition/access to the sport; a global pandemic that suspended multiple tournaments in 2020) relate to the significance of these and other performance statistics?

\subsection{Contributions}
The Graf-Williams example highlights an inherent difficulty with the problem of identifying the GOAT: individual performance statistics can only offer snapshots of history, and it is challenging to objectively contextualize them over time.
Fortunately, most individual sports maintain regularly-updated \emph{rankings} that reflect the moment in time in which they were curated.
In professional tennis specifically, rankings are curated on a weekly basis by both the Association of Tennis Professionals (ATP)~\cite{atp_rulebook_2024} and the Women's Tennis Association (WTA)~\cite{wta_rulebook_2024}.

In this work, we propose a mathematical framework that leverages historical ranking data to \emph{compare} players, and in particular players spanning across different time periods in sports history.
We summarize our method in Figure~\ref{fig: summary}.
In broad strokes, it is as follows:
\begin{enumerate}
    \item 
    First, we use cutoffs of historical ranking data (Figure~\ref{fig: historical}) to design a data-driven random walk on the symmetric group, from which we obtain a stationary distribution over rankings on the full set of players (Figure~\ref{fig: walk}).
    \item
    Next, we combine this distribution with a notion of stochastic dominance to obtain a partially ordered set (poset) of players (Figure~\ref{fig: poset}).
\end{enumerate}

\begin{figure}[ht]
    \centering
    \begin{subfigure}{.32\linewidth}
      \centering
      \includegraphics[width=\linewidth]{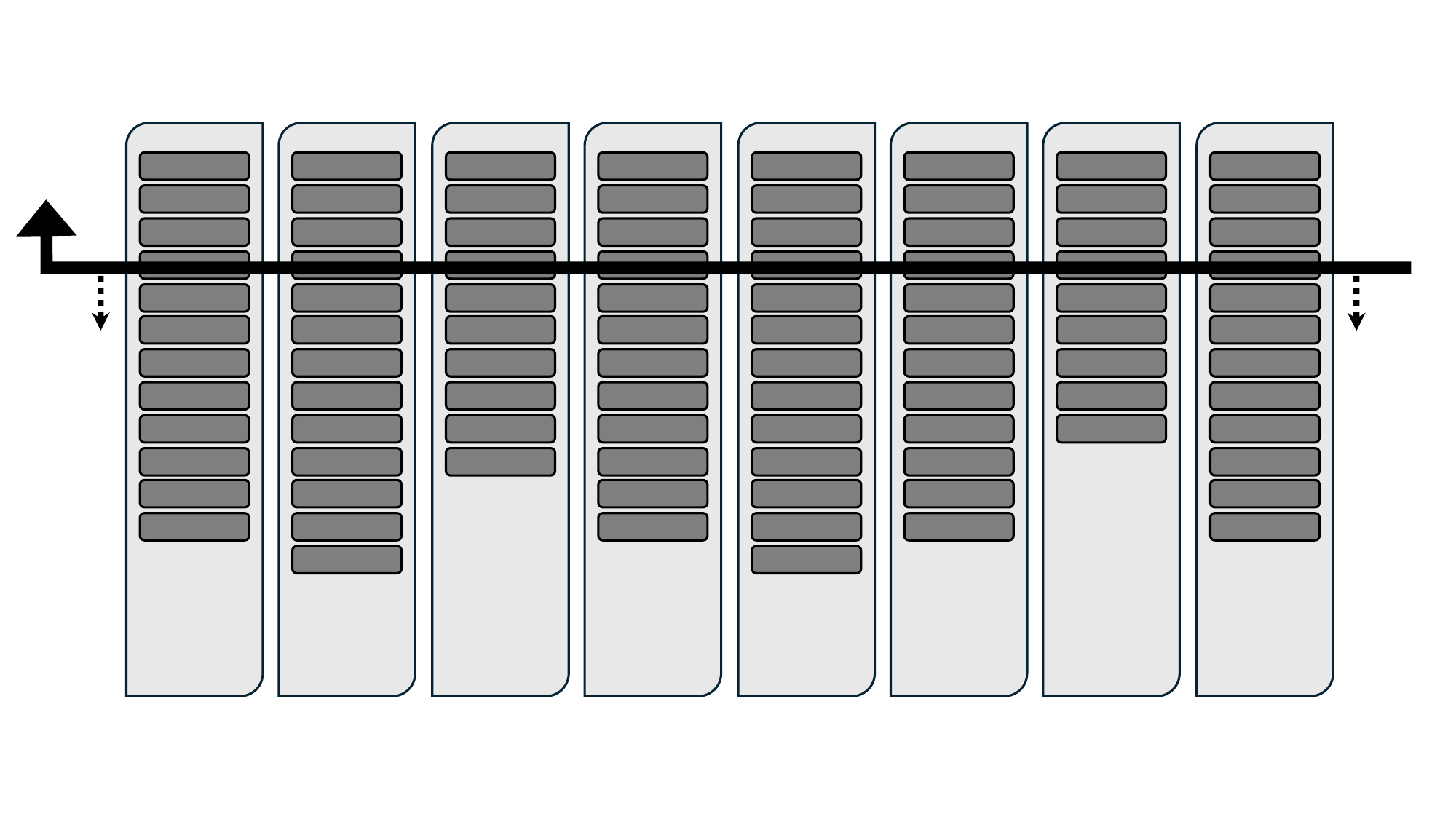}
      \caption{Historical rankings.}
      \label{fig: historical}
    \end{subfigure}
    \hfill
    \begin{subfigure}{.32\linewidth}
      \centering
      \includegraphics[width=\linewidth]{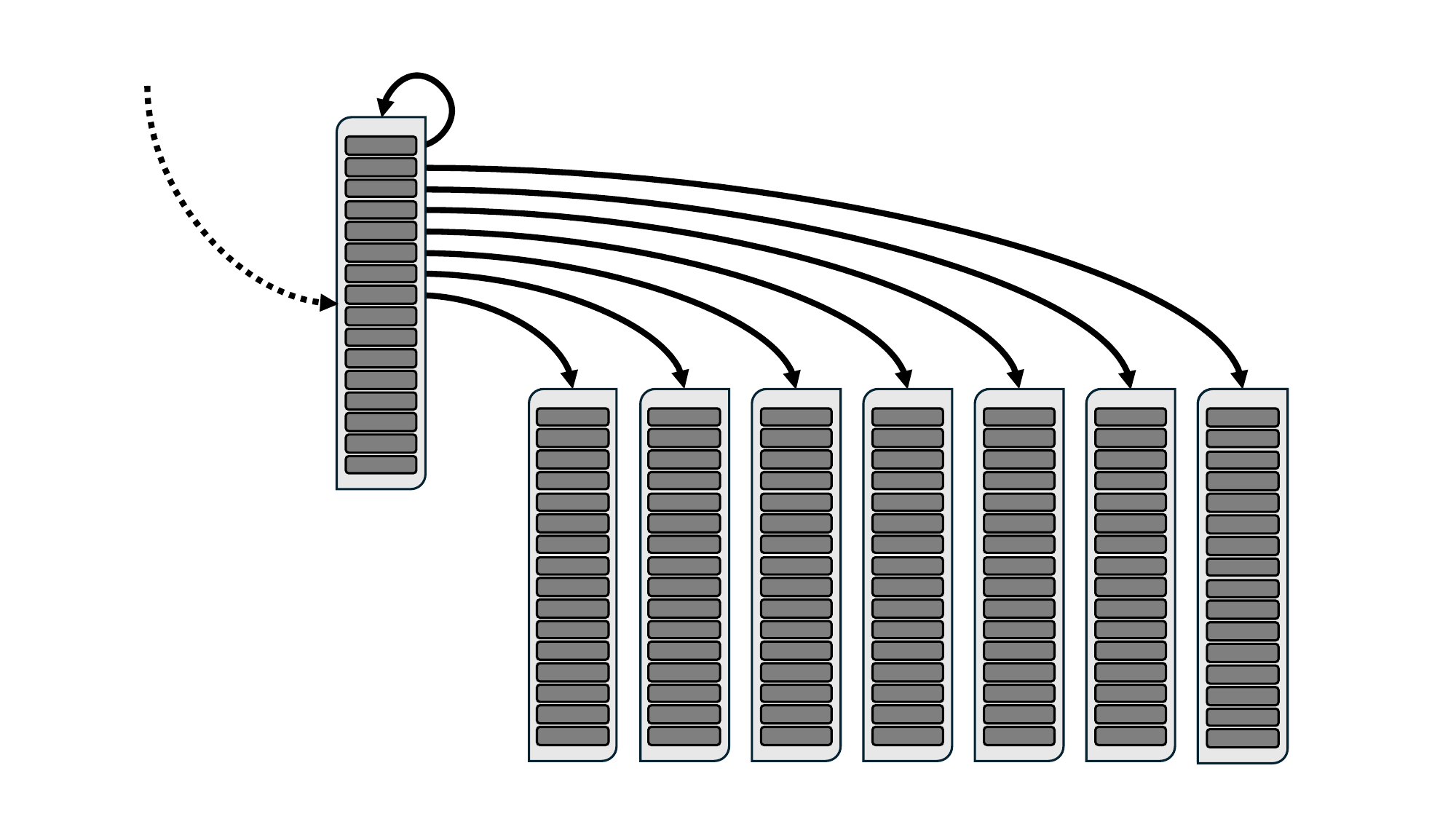}
      \caption{Data-driven random walk.}
      \label{fig: walk}
    \end{subfigure}
    \hfill
    \begin{subfigure}{.32\linewidth}
      \centering
      \includegraphics[width=\linewidth]{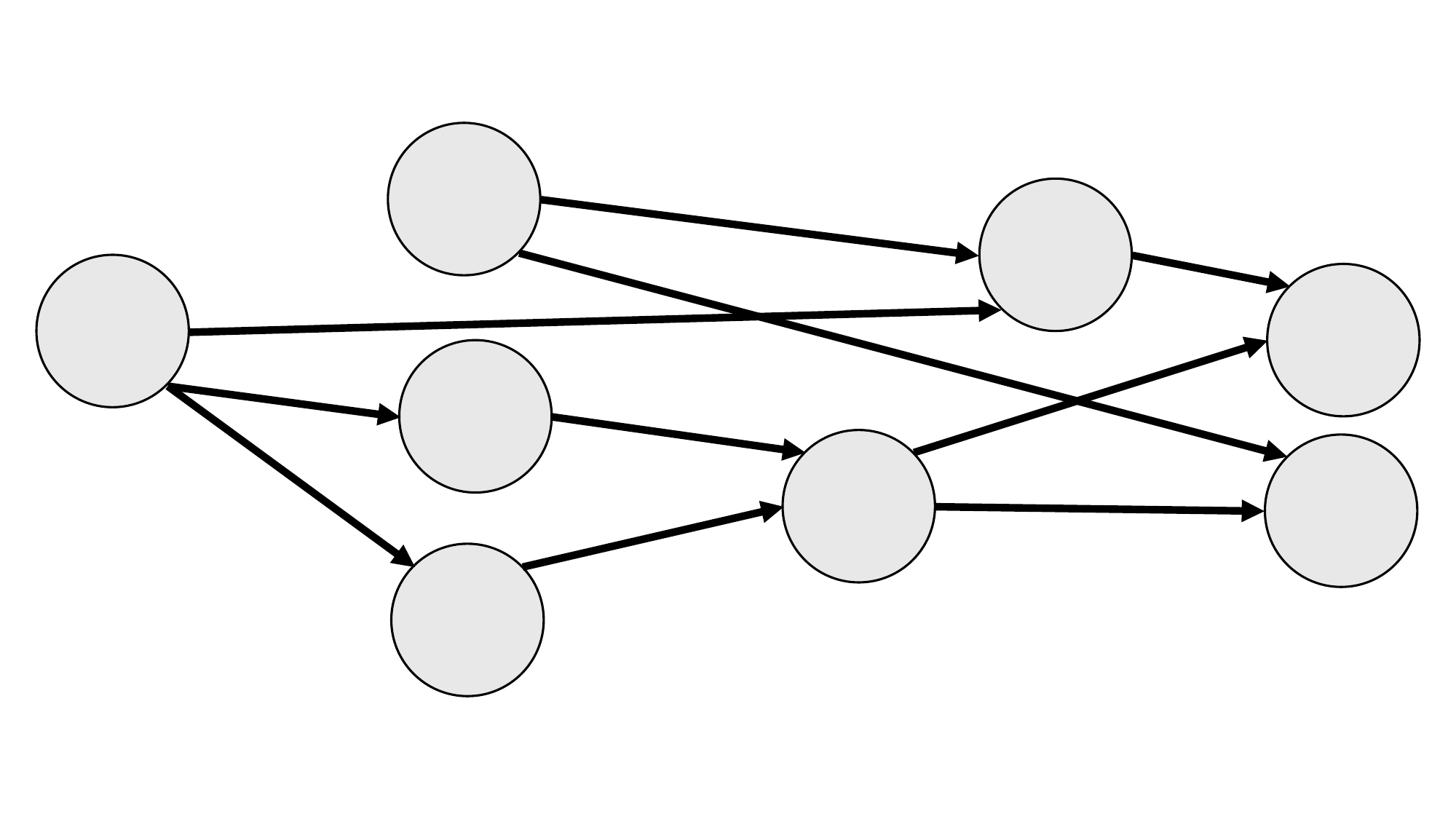}
      \caption{Poset over players.}
      \label{fig: poset}
    \end{subfigure}
    \caption{
        Outline of our method. 
        Figure~\ref{fig: historical}: The input data consists of the top portions of the full historical rankings, at varying cutoffs.
        Figure~\ref{fig: walk}: The random walk is designed to have a stationary distribution over rankings on the full set of players, reflecting patterns of relative order present in the input data.
        Figure~\ref{fig: poset}:
        The output data consists of a poset over the full set of players, based on a notion of stochastic dominance.
    }
    \label{fig: summary}
\end{figure}

This approach to player comparison represents a departure from existing methods in two principal ways.
First, compared to the use of individual performance statistics or match-level data, the use of historical rankings more directly reflects the evolution of circumstances and/or value systems.
For example, players might be strategic around tournament participation in response to the ranking system in place\textemdash what one had to do to be top-ranked in the 1980s might be different from what one had to do in the 2010s\textemdash and this might not be as evident in other datasets.
Moreover, it opens the opportunity for longitudinal comparisons in ranked systems for which match-level data is not available/does not exist (e.g., individual players within team sports, music charts, college rankings).
Second, the adoption of a notion of stochastic dominance helps formally come to terms with the possibility that there is no single GOAT, instead limiting pairwise comparisons to cases in which these can be established conclusively.

To illustrate the practical merits of our approach, we perform a numerical implementation using singles ranking data from the ``Open Era'' of the ATP and the WTA.
We experiment with a concept of GOAT that captures individual performance at the top of ranking cutoffs\textemdash performance as a ``top 3,'' ``top 5,'' ``top 10,'' or ``top 20'' player (refer to Figure~\ref{fig: historical} and Section~\ref{sec: cutoffs and the meaning of GOAT} for details)\textemdash to obtain data-driven posets and find the GOATs in the respective categories.

Tying back to the introduction, we conclude that Graf and Williams are both consistently the GOATs.
In particular, they are distinct maximal poset elements in all corresponding experiments.
We summarize these and other main findings about the GOATs in the ATP and the WTA in Tables~\ref{tab: atp} and~\ref{tab: wta}, respectively.
\begin{table}[ht]
\centering
\resizebox{\columnwidth}{!}{%
\begin{tabular}{l|l|l}
\textbf{Cutoff} & \textbf{GOATs of the ATP}                      & \textbf{Notes}                                                                                                      \\ \hline
Top 3  & Djokovic, Federer, Nadal, and Sampras & - All better than Connors and Lendl                                                                         \\ \hline
Top 5 &
  Djokovic, Federer, and Nadal &
  \begin{tabular}[c]{@{}l@{}}- Federer and Nadal better than Sampras and Connors \\ - Connors, Djokovic, Lendl, and Sampras incomparable\end{tabular} \\ \hline
Top 10 &
  Djokovic, Federer, and Nadal &
  \begin{tabular}[c]{@{}l@{}}- All better than Connors and Sampras\\ - Federer and Nadal better than Becker and Lendl \\ - Becker, Djokovic, and Lendl incomparable\end{tabular} \\ \hline
Top 20 & Federer, Nadal, and Lendl                     & \begin{tabular}[c]{@{}l@{}}- All better than Djokovic\\ - Djokovic better than Sampras \end{tabular} \\ \hline
\end{tabular}%
}
\caption{
Summary of GOATs of the ATP. 
The surnames mentioned correspond to Boris Becker, Jimmy Connors, Novak Djokovic, Roger Federer, Ivan Lendl, Rafael Nadal, and Pete Sampras (in alphabetical order).
Refer to Section~\ref{sec: rankings from the atp} for more details.
}
\label{tab: atp}
\end{table}
\begin{table}[ht]
\centering
\resizebox{\columnwidth}{!}{%
\begin{tabular}{l|l|l}
\textbf{Cutoff} &
  \textbf{GOATs of the WTA} &
  \textbf{Notes} \\ \hline
Top 3 &
  Graf and S. Williams &
  \begin{tabular}[c]{@{}l@{}}- Graf better than Hingis and Navratilova \\ - S. Williams better than Hingis \\ - S. Willams and Navratilova incomparable \end{tabular} \\ \hline
Top 5 &
  Graf and S. Williams &
  \begin{tabular}[c]{@{}l@{}}- Graf better than Navratilova \\ - Graf and S. Williams better than Seles and Hingis \\ - S. Williams and Navratilova incomparable \end{tabular} \\ \hline
Top 10 &
  Graf, Navratilova, and S. Williams &
  \begin{tabular}[c]{@{}l@{}}- Graf and Navratilova better than Seles \\ - S. Williams and Seles incomparable \\ - All better than Hingis \end{tabular} \\ \hline
Top 20 &
  \begin{tabular}[c]{@{}l@{}}Graf, Navratilova, S. Williams, \\ V. Williams, and Wozniacki\end{tabular} &
  - Numerous players close to maximal \\ \hline
\end{tabular}%
}
\caption{
Summary of GOATs of the WTA. 
The surnames mentioned correspond to Steffi Graf, Martina Hingis, Martina Navratilova, Monica Seles, Serena Williams, Venus Williams, and Caroline Wozniacki (in alphabetical order).
Refer to Section~\ref{sec: rankings from the wta} for more details.}
\label{tab: wta}
\end{table}
In addition, we observe that prominent young players such as Carlos Alcaraz in the ATP and Aryna Sabalenka in the WTA perform well as top 3 players, but that they are at a disadvantage as top 5, 10, or 20 players due to their short career lengths at the time of writing.

Lastly, we analyze the obtained posets and note qualitative differences between the rankings of the two associations.
We find that, in the ATP, being a ``good top 20 player'' is correlated with being a good top 10, 5, and 3 player (with decreasing correlation in that order).
However, in the WTA, \emph{being a ``good top 20 player'' is only mildly correlated with being a good top 10 player, and it is nearly uncorrelated with being a good top 5 or top 3 player.} 
This highlights a measurable difference in terms of career and dominance longevity for top players across the ATP and the WTA.

\subsection{Related Work}
\label{sec: related work}
Radicchi~\cite{radicchi2011best} addressed the question of who is the GOAT in men's professional tennis using historical match-level data and a network diffucion algorithm similar to \texttt{PageRank}~\cite{brin1998anatomy}.
Baker and McHale~\cite{baker2014dynamic} addressed the same question using time-varying paired comparison models and the results of matches in major tournaments.
Temesi, Sz{\'a}doczki, and Boz{\'o}ki~\cite{temesi2024incomplete} addressed the question of who is the GOAT in women's professional tennis using pairwise comparison matrices and the results of matches between players who were ever ranked first.
Still, from a philosophical point of view, sports rankings have been criticized for their inability to capture comparisons that are ``too close to call''~\cite{bordner2016all}.

In contrast to existing methods, we propose the aggregation of historical rankings into a poset as an alternative to the stringent requirements of a linear order.
That is, some pairs of players can be conclusively compared, others cannot, and there might be various, pairwise incomparable GOATs.
Thus, our conceptualization of the problem is most closely related to that of \emph{rank aggregation}~\cite{dwork2001rank,schalekamp2009rank} (refer also to Lin~\cite{lin2010rank} for a survey).

However, our methods are actually most closely related to those from the theory of card shuffling\textemdash the design and analysis of methods to sample uniform random permutations of a fixed finite set.
The set $[n] \coloneqq \{1, 2, \ldots, n\}$ corresponds to a deck of $n \in \mathbb{N} \coloneqq \{1, 2, \ldots\}$ cards, and different card shuffling procedures (e.g., riffle shuffling, shuffling by transposing two random cards, shuffling by transposing two random adjacent cards) correspond to different sampling methods.
The central questions in this area relate to convergence rates: how many times does a given shuffling procedure need to be repeated so that a given deck is sufficiently close to uniformly distributed?
Bayer and Diaconis~\cite{bayer1992trailing} showed that $(3/2) \log n + \epsilon$ riffle shuffles are necessary and sufficient for a deck of $n$ cards to start becoming well-mixed, exhibiting a cutoff phenomenon~\cite{aldous1986shuffling} (so that a standard deck of $52$ cars requires about $6$ repetitions).
Diaconis and Shahshahani~\cite{diaconis1981generating} showed that the random transposition shuffle has a cutoff at $(1/2) n \log n$, whereas Lacoin~\cite{lacoin2016mixing} showed that the random adjacent transposition shuffle has a cutoff at $(1/\pi^2) n^3 \log n$.
Refer to Diaconis and Fulman~\cite{diaconis2023mathematics} for more on card shuffling.

Our method in Section~\ref{sec: data-driven random walks on the symmetric group} is most similar to shuffling by random (not necessarily adjacent) transpositions.
However, our goal in this work is different from card shuffling in that we do not seek to sample uniform random permutations.
Rather, we seek to sample permutations from a data-driven distribution that is somehow in agreement with historical ranking data. 
Given the large amount of existing work for the case of uniform random permutations, we anticipate analytical questions around the data-driven sampling of permutations to be challenging and of independent mathematical interest.

\subsection{Organization}
\label{sec: organization}
The remainder of this work is organized as follows.
For the purpose of self-containment, in Section~\ref{sec: background} we summarize some necessary technical background.
In Section~\ref{sec: data-driven random walks on the symmetric group} we design our methods.
In Section~\ref{sec: data-driven random walks on the symmetric group} we implement our methods using data from the ATP and the WTA, and analyze the results of our numerical implementation.
Lastly, in Section~\ref{sec: conclusion} we make concluding remarks and point out other potential applications of our methods.

\section{Background}
\label{sec: background}
Readers familiar with the symmetric group may choose to skip Section~\ref{sec: symmetric group}, whereas readers familiar with Markov chains may choose to skip Section~\ref{sec: markov chains}.
In Section~\ref{sec: sports rankings and partially ordered sets} we formalize our problem input along with other necessary technical concepts and notation.

\subsection{Symmetric Group}
\label{sec: symmetric group}
Let $n \in \N \coloneqq \{1, 2, 3, \ldots\}$.
A permutation of $[n] \coloneqq \{1, 2, \ldots, n\}$ is a bijection $\pi: [n] \rightarrow [n]$.
The \emph{symmetric group} of degree $n$, denoted $\mathfrak{S}_n$, is the set of all permutations of $[n]$ forming a group under function composition.
\begin{fact}
\label{fact: n!}
For $n \geq 1$, $|\mathfrak{S}_n| = n!$
\end{fact}

The following example illustrates the different ways in which we express permutations.
\begin{example}
\label{ex: one-line}
Let $\pi : [5] \rightarrow [5]$ be the permutation expressed explicitly as $\pi(1) = 5$, $\pi(2) = 4$, $\pi(3) = 1$, $\pi(4) = 2$, and $\pi(5) = 3$.
Equivalently, we express $\pi$ as follows:
\begin{itemize}
    \item 
    In one-line notation, and write $\pi = 54123$.
    Here, the $i$th entry of the word is $\pi(i)$.
    \item
    In cycle notation, and write $\pi = (153)(24)$.
    Here, the number following number $i$ in a cycle is $\pi(i)$.
    For example, $5$ follows $1$ in $(153)$ so that $\pi(1) = 5$.
    Similarly, $1$ follows $3$ in $(153)$ so that $\pi(3) = 1$.
    Note that function composition is performed from right to left, so that $(153) = (13)(15)$ and equivalently $\pi = (13)(15)(24)$.
\end{itemize}
\end{example}
Cycles of length $2$ are known as \emph{transpositions}, since the effect of a cycle $(i,j)$ for distinct $i, j \in [n]$ is to transpose the $i$th and $j$th entries.
By convention, a transposition $(i,j)$ is written as $(i, j)$ if $i < j$ and as $(j, i)$ if $j < i$.

\begin{fact}
\label{fact: product of 2-cycles}
For $n > 1$, every $\pi \in \mathfrak{S}_n$ is a product of at most $n-1$ transpositions.
\end{fact}

Moreover, every transposition can be expressed as a product of \emph{adjacent transpositions}.
For ease of notation, let $s_{ij} = (i, j)$ denote a transposition and $s_i = (i, i + 1)$ for $i \in [n-1]$ denote an adjacent transposition.
Consider the following example.
\begin{example}
Note that
\begin{itemize}
    \item
    $s_{13} = (13) = (12)(23)(12) = s_1 s_2 s_1$,
    \item 
    $s_{15} = (15) = (12)(23)(34)(45)(34)(23)(12) = s_1 s_2 s_3 s_4 s_3 s_2 s_1$, and
    \item 
    $s_{24} = (24) = (23)(34)(23) = s_2 s_3 s_2$.
\end{itemize}
Therefore, 
\begin{align*}
    \pi
    &= (153)(24) \\
    &= (13)(15)(24) \\
    &= (12)(23)(12)(12)(23)(34)(45)(34)(23)(12)(23)(34)(23) \\
    &= (12)(34)(45)(34)(23)(12)(23)(34)(23) \\
    &= s_1 s_3 s_4 s_3 s_2 s_1 s_2 s_3 s_2.
\end{align*}
\end{example}

The last example illustrates the following.
\begin{fact}
\label{fact: product of adjacent 2-cycles}
For $n > 1$, every $\pi \in \mathfrak{S}_n$ is a product of at most $n(n-1)/2$ adjacent transpositions.
\end{fact}

In this work, we associate with $\mathfrak{S}_n$ two bi-directed graphs over $n!$ nodes (the number of nodes follows from Fact~\ref{fact: n!}):
\begin{enumerate}
    \item 
    The first one, depicted in Figure~\ref{fig: transpositions} for $n=4$, contains arcs $(\pi, \sigma)$ and $(\sigma, \pi)$ if and only if $\pi = \sigma \cdot s_{ij}$ for some $1 \leq i < j \leq n$ .
    For example, if $\pi = 4213 \in \mathfrak{S}_4$ and $\sigma = 4312 \in \mathfrak{S}_4$, then the graph has the arc $(\pi, \sigma)$ since $4213 \cdot s_{24} = 4312$ and $s_{24}$ is a transposition. 
    The undirected version of this bi-directed graph is known as the Bruhat graph.
    \item
    The second one, depicted in Figure~\ref{fig: adjacent transpositions} for $n=4$, contains arcs $(\pi, \sigma)$ and $(\sigma, \pi)$ if and only if $\pi = \sigma \cdot s_{i}$ for some $i \in [n - 1]$.
    For example, if $\pi = 4213 \in \mathfrak{S}_4$ and $\sigma = 4231 \in \mathfrak{S}_4$, then the graph has the edge $(\pi, \sigma)$ since $4213 \cdot s_3 = 4231$ and $s_3$ is an adjacent transposition. 
    The undirected version of this bi-directed graph is realized as the $1$-skeleton graph of the permutahedron.
\end{enumerate}
\begin{figure}[ht]
    \centering
    \begin{subfigure}{.4975\linewidth}
      \centering
      \includegraphics[width=\linewidth]{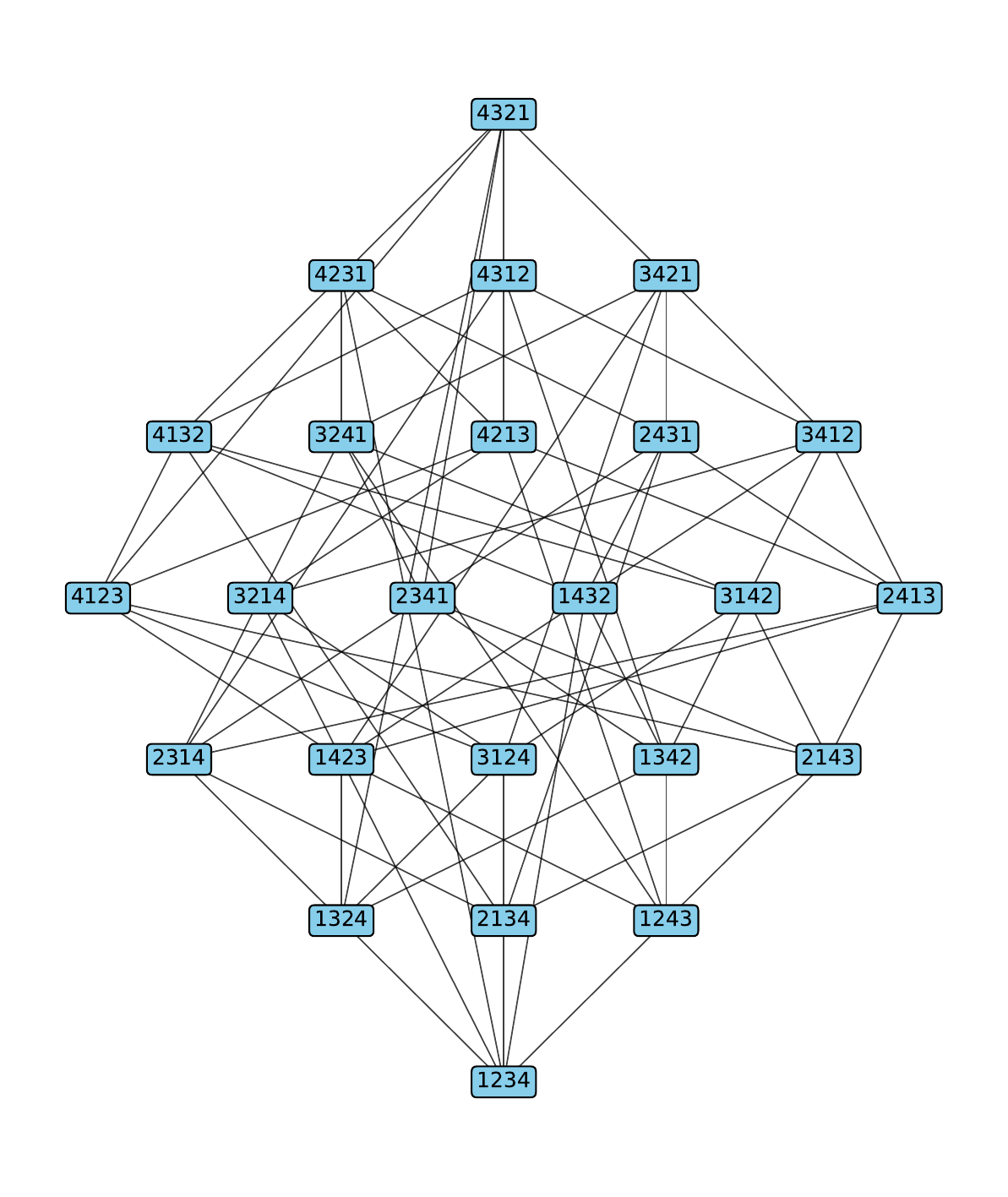}
      \caption{Generated by transpositions.}
      \label{fig: transpositions}
    \end{subfigure}%
    \begin{subfigure}{.4975\linewidth}
      \centering
      \includegraphics[width=\linewidth]{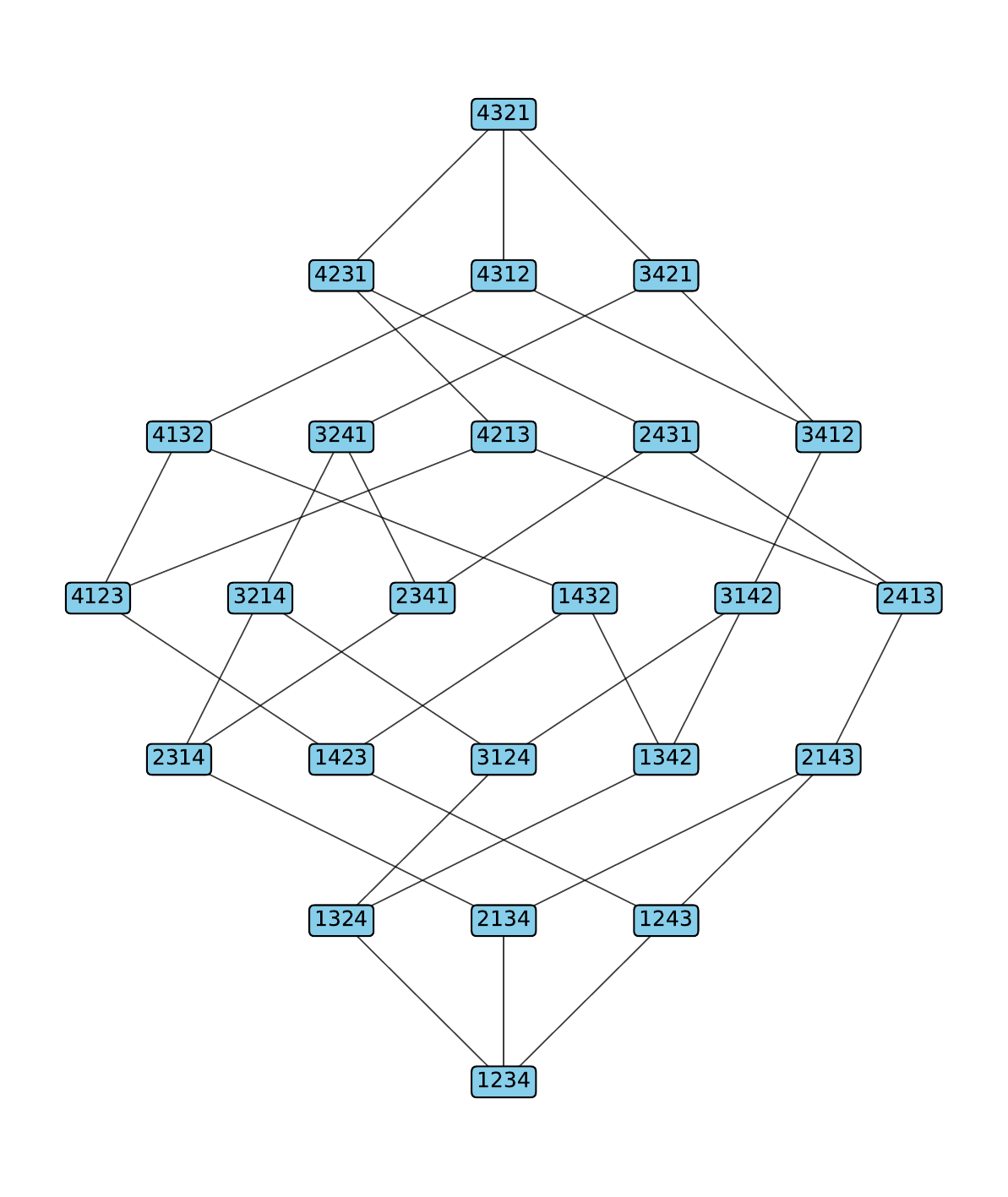}
      \caption{Generated by adjacent transpositions.}
      \label{fig: adjacent transpositions}
    \end{subfigure}
    \caption{
        Bi-directed graphs on $\mathfrak{S}_4$, with nodes as permutations in one-line notation.
        (Bi-directed arcs are depicted as undirected edges.)
        }
    \label{fig: graphs}
\end{figure}

The bi-directed graph in Figure~\ref{fig: adjacent transpositions} illustrates the commonly-used fact that adjacent transpositions suffice to generate $\mathfrak{S}_n$.
However, as we explain in Section~\ref{sec: data-driven random walks on the symmetric group}, the more edge-dense class of graphs illustrated in Figure~\ref{fig: transpositions} is better suited for the purposes of this work.
The main properties of this class of graphs that we leverage in this work follow from Fact~\ref{fact: product of 2-cycles} and are summarized below.
\begin{remark}
\label{remark: g_n}
For $n > 1$, the bi-directed version of the Bruhat graph on $\mathfrak{S}_n$ is strongly-connected and has diameter $n-1$.
\end{remark}
Refer to any introductory abstract algebra textbook for more on the symmetric group.

\subsection{Markov Chains}
\label{sec: markov chains}

This section is based on Levin and Peres~\cite{levin2017markov}.
A discrete-time Markov chain is a sequence $X_1, X_2, X_3, \ldots$ of random variables with the Markov property.
That is,
\begin{equation*}
    \Pr[X_{t+1} = x | X_1 = x_1, X_2 = x_2, \ldots, X_t = x_t] = \Pr[X_{t+1} = x | X_t = x_t] 
\end{equation*}
for all $t \in \mathbb{N}$ whenever the conditional probabilities are well-defined.
If the random variables share a finite co-domain, then the chain has a finite state space. 
Say, it is indexed by $[m]$ for some $m \in \mathbb{N}$.
The chain is time-homogeneous if
\begin{equation*}
    \Pr[X_{t+1} = x | X_t = y] = \Pr[X_t = x | X_{t-1} = y] 
\end{equation*}
for all $t \in \mathbb{N}$, so that the transition probabilities are independent of $t$.

In the finite and time-homogeneous case, the transition probabilities can be represented by a row-stochastic matrix $P \in \mathbb{R}^{m \times m}$ where
\begin{equation*}
    P_{ij} = \Pr[X_{t+1} = j | X_t = i]
\end{equation*}
for all $i, j \in [m]$.
In effect, $P$ encodes the transition probabilities of a random walk on a directed graph $G_m = ([m], A)$, where $(i,j) \in A$ if and only if $P_{ij} > 0$.
Specifically, $X_t \in [m]$ corresponds to the node visited during the $t$th step and, conditioned on $X_t = i$, the $i$th row of $P$ governs the random transition to the $(t+1)$th step.

The chain $P$ is \emph{irreducible} if, for any $i, j \in [m]$, there exists some $\tau \in \mathbb{N}$ such that $P_{ij}^\tau > 0$.
For any $i \in [m]$, its \emph{period} is defined as the greatest common divisor of $\{t \in \mathbb{N}: P_{ii}^t > 0\}$.
The chain is \emph{aperiodic} if $\text{gcd}\left(\{t \in \mathbb{N}: P_{ii}^t > 0\}\right) = 1$ for all $i \in [m]$. 
Note that $P_{ii} > 0$ for all $i \in [m]$ is a sufficient condition for aperiodicity.

The following are classical results in the theory of Markov chains.
\begin{theorem}
\label{theorem: stationary}
If $P$ is an irreducible finite Markov chain, then there exists a unique probability distribution $\mu \in \mathbb{R}^m$ over $[m]$ with $\mu = \mu P$.
The distribution $\mu$ is referred to as the \emph{stationary distribution} of $P$. 
\end{theorem}
\begin{theorem}[Convergence Theorem]
\label{theorem: convergence}
If $P$ is an irreducible and aperiodic finite Markov chain with stationary distribution $\mu$, then there exist constants $\alpha \in (0, 1)$ and $C > 0$ such that
\begin{equation*}
    \max_{i \in [m]} \lVert {(P^t)}_i - \mu \rVert_{\mathrm{TV}} \leq C\alpha^t,
\end{equation*}
where ${(P^t)}_i$ is the $i$th row of $P^t$ and $\lVert \cdot \rVert_{\mathrm{TV}}$ is the total variation distance.
\end{theorem}
Theorem~\ref{theorem: convergence} gives analytical grounds for the repeated application of an irreducible and aperiodic finite Markov chain to sample from its stationary distribution.

\subsection{Sports Rankings and Partially Ordered Sets}
\label{sec: sports rankings and partially ordered sets}

A finite \emph{partially ordered set} (poset) is a pair $\pi = ([n], \preceq_{\pi})$ consisting of a ground set $[n]$ and a relation $\preceq_{\pi}$ satisfying:
\begin{enumerate}
    \item \emph{Reflexivity:} For any $i \in [n]$ we have that $i \preceq_{\pi} i$.
    \item \emph{Antisymmetry:} For any $i, j \in [n]$ we have that if $i \preceq_{\pi} j$ and $j \preceq_{\pi} i$, then $i = j$.
    \item \emph{Transitivity:} For any $i, j, k \in [n]$ we have that if $i \preceq_{\pi} j$ and $j \preceq_{\pi} k$, then $i \preceq_{\pi} k$.
\end{enumerate}
For $i, j \in [n]$, let $i \prec_\pi j$ if $i \preceq_\pi j$ and $i \neq j$.
An element $i \in [n]$ is \emph{maximal} if there is no $j \neq i \in [n]$ with $i \prec_\pi j$.
The definition of \emph{minimal} is analogous.

A poset may be represented pictorially by a Hasse diagram, which is a drawing of its transitive reduction.
Formally, for any $i, j \in [n]$, we say that $j$ covers $i$ if $i \prec_\pi j$ and there is no $k \in [n]$ satisfying $i \prec_\pi k \prec_\pi j$.
Then, the Hasse diagram of $\pi$ is a directed graph whose vertices are $[n]$ and whose edges are given by the cover relations.
For example, if the edges in the graphs in Figure~\ref{fig: graphs} were all oriented upwards, we would obtain Hasse diagrams of two different posets on the symmetric group of order $4$, both with unique minimal element $1234$ and unique maximal element $4321$.

A totally ordered set is a poset in which every two elements are comparable, and a chain is a subset of a poset that is totally ordered for the induced order (note that the use of the word ``chain'' in the context of a poset differs from its use in Section~\ref{sec: markov chains}).
For example, $\{1234, 1243, 2143, 2413, 2431, 4231, 4321\}$ forms a chain in Figure~\ref{fig: adjacent transpositions}.
An \emph{anti-chain} is a subset of elements such that no pair of elements are comparable.
For example, $\{1324, 2134, 1243\}$ forms an anti-chain in Figure~\ref{fig: adjacent transpositions}.
A poset $\pi^* = ([n], \preceq_{\pi^*})$ is an extension of $\pi = ([n], \preceq_{\pi})$ if, for any $i, j \in [n]$ with $i \preceq_\pi j$, we have that $i \preceq_{\pi^*} j$ as well.
An extension is a linear extension if it is totally ordered.
Refer to Stanley~\cite[Chapter~3]{stanley2012enumerative} for more on posets.

Now, let $[n]$ represent a set of players.
A \emph{ranking} is a totally ordered set $\pi = ([n], \preceq_{\pi})$, which we treat as a permutation $\pi \in \mathfrak{S}_n$ (in explicit notation) given the natural ordering $\leq$ of $\mathbb{N}$.
In particular, we have that
\begin{equation*}
    j \preceq_\pi i \iff \pi(i) \leq \pi(j)
\end{equation*}
for all $i, j \in [n]$.
We more generally extend the notion of ranking to proper subsets of players, so that a ranking is a totally ordered set $r = (S_r, \preceq_r)$ for some $S_r \subseteq [n]$, which we again treat as a bijection $r: S_r \rightarrow \left[|S_r|\right]$ given the natural ordering $\leq$ of $\mathbb{N}$.
In this way, our problem input is a collection 
\begin{equation*}
    \mathcal{R} = \{r_1, r_2, \ldots, r_\ell\}    
\end{equation*}
of historical sports rankings on subsets of players (refer to Figure~\ref{fig: historical}).
In subsequent sections we will consider partial orders on the full set of players $[n]$ together with their corresponding linear extensions as part of our problem output (refer to Figure~\ref{fig: poset}).

Lastly, note that we adopt the following notation.
For $i \in [n]$, let $\mathcal{R}_{i} = \{r \in \mathcal{R}: i \in S_r\}$ denote the sub-collection of rankings containing player $i$.
Similarly, for distinct $i, j \in [n]$, let $\mathcal{R}_{ij} = \{r \in \mathcal{R}: i, j \in S_r\}$ denote the sub-collection of rankings containing both player $i$ and player $j$.
By convention, let $\mathcal{R}_{ii} = \emptyset$.

\section{Data-Driven Random Walks on the Symmetric Group}
\label{sec: data-driven random walks on the symmetric group}

In Section~\ref{sec: cutoffs and the meaning of GOAT} we discuss the meaning of GOAT in relation to assumptions (in particular, ranking cutoffs) on the problem input.
In Section~\ref{sec: transition rule} we design a data-driven transition rule for random walks on $\mathfrak{S}_n$.
In Section~\ref{sec: partial ordering} we adopt a notion of stochastic dominance to build a poset from the stationary distribution of this random walk.
In Section~\ref{sec: sampling} we discuss sampling procedures for this random walk.

\subsection{Cutoffs and the meaning of GOAT}
\label{sec: cutoffs and the meaning of GOAT}

Recall that $[n]$ represents the set of players and that the problem input is a collection 
\begin{equation*}
    \mathcal{R} = \{r_1, r_2, \ldots, r_\ell\}    
\end{equation*}
of historical sports rankings on subsets of the players, so that $\cup_{r \in \mathcal{R}} S_r = [n]$.
Now, consider any player $i \in [n]$ and any ranking $r \in \mathcal{R}_i$.
Certainly, player $i$ being ``top ranked'' in $r$ is consistent with them being the GOAT, but how ``top'' is top ranked?
A full sports ranking might involve thousands of players.
However, player $i$ being among the top $10$ in $r$ is much more indicative than them being among the top $100$, which is itself even much more indicative than them being among the top $1000$.
In some sense, we must adopt a definition for the concept of being at the top.

In the experimental part of this work, we restrict the rankings in $\mathcal{R}$ to involve the top $3$, $5$, $10$, and $20$ players as cutoffs from the full ranking data that might be available.
This is illustrated in Figure~\ref{fig: historical}.
We do this for two reasons:
\begin{enumerate}
    \item 
    The first reason is conceptual.
    As argued previously, the difference between a player being ranked $100$ or $101$ (let alone $1000$ or $1001$) is not particularly relevant in relation to them being the GOAT, whereas the difference between them being ranked $1$ or $2$ very much is.
    By experimenting with different cutoffs we gain insights into the consequences of this subjective choice.
    \item 
    The second reason is computational.
    Indeed, if each ranking $r \in \mathcal{R}$ involved thousands of players, the set $[n]$ would be very large and $n!$ would be astronomically large.
    In our experiments, the different cutoffs always keep the size of the set under $300$ players.
    Note that $300!$ is still an astronomically large number, with $615$ decimal digits.
\end{enumerate}
In short, we assume that being the GOAT is consistent with being a ``good top $\kappa$'' player, and we experiment with the choice of $\kappa \in \mathbb{N}$.
In this way, having ever made it to the top $\kappa$ is a necessary condition for being the GOAT,
which restricts the set of contenders.

\subsection{Transition Rule}
\label{sec: transition rule}

Let $W \in \mathbb{R}_{\geq 0}^{n \times n}$ be a matrix with entries
\begin{equation}
\label{eq: W}
    W_{ij} 
    = 
    1 + \underbrace{\left|\mathcal{R}_i\right| / \left|\mathcal{R}_j\right|}_{\text{prevalence ratio}} \cdot \ \underbrace{\left|\{r \in \mathcal{R}_{ij}: j \prec_{r} i\}\right|}_{\text{$\#$ of times $i$ beats $j$}}
\end{equation}
for all distinct $i, j \in [n]$.
The second term in the product, $\left|\{r \in \mathcal{R}_{ij}: j \prec_{r} i\}\right|$, encodes the number of cut off historical rankings in which player $i$ is ranked better than player $j$.
We scale this quantity by the first term in the product, $\left|\mathcal{R}_i\right| / \left|\mathcal{R}_j\right|$, which is the players' relative prevalence above the cutoff.
The motivation for this is as follows.
\begin{remark}
Suppose players $i$ and $j$ have short and long prevalence in the cut off historical rankings, respectively.
In particular, it might be that player $i$ made it to the top of the ranking for a brief period of time, consequently beating player $j$ during this period.
However, since the prevalence of player $i$ is short, there might be significant periods of time during which player $i$ fell below the cutoff while player $j$ did not\textemdash by definition, player $j$ beat player $i$ during these periods.
This artifact of cutoffs is not captured by the term $\left|\{r \in \mathcal{R}_{ij}: j \prec_{r} i\}\right|$, which by itself favors player $i$ over player $j$.
We correct for this in the model using their prevalence ratio $\left|\mathcal{R}_i\right| / \left|\mathcal{R}_j\right|$.
\end{remark}
Intuitively, the term $\left|\mathcal{R}_i\right| / \left|\mathcal{R}_j\right|$ ensures that if the players have similar prevalence above the cutoff, the data points in which player $i$ beats player $j$ are taken at face value.
Conversely, if the prevalence of player $i$ is much smaller than that of player $j$, the data points in which player $i$ beats player $j$ are proportionally discounted.
Finally, the additive unit is included to later on obtain an irreducible random walk.
Note that \eqref{eq: W} can be pre-computed efficiently.

Next, for each $\pi \in \mathfrak{S}_n$ and each $1 \leq i < j \leq n$, let
\begin{equation}
\label{eq: W(s_ij | pi)}
    \mathcal{W}(s_{ij} | \pi) \propto 
    \begin{cases}
        W_{ji}, & \text{ if } j \prec_\pi i \\
        W_{ij}, & \text{ otherwise }
    \end{cases}
\end{equation}
form a probability distribution over (not necessarily adjacent) transpositions.
The interpretation of $\mathcal{W}(s_{ij} | \pi)$ is as follows:
\begin{itemize}
    \item $j \prec_\pi i \implies i \prec_{s_{ij} \cdot \pi} j$, whereas
    \item $i \prec_\pi j \implies j \prec_{s_{ij} \cdot \pi} i$.
\end{itemize}
In other words, the probability of selecting transposition $s_{ij}$ given $\pi$ is proportional to the agreement between the resulting permutation $s_{ij} \cdot \pi$ and the historical data, particularly with respect to the relative order between players $i$ and $j$.

\begin{remark}
\label{remarl: non-adjacent}
Note that \eqref{eq: W}-\eqref{eq: W(s_ij | pi)} highlight the need to consider the set of all transpositions, not just those that are adjacent (recall the difference in Figure~\ref{fig: graphs}).
Indeed, depending on the structure of $\mathcal{R}$ and the indexing of the players, it is possible that
\begin{equation*}
    \left|\mathcal{R}_{ij}\right| = 0
\end{equation*}
for a large fraction of the adjacent choices of $1 \leq i < j (= i + 1) \leq n$, in turn rendering \eqref{eq: W} uninformative.
For example, this could occur if the players are indexed in such a way that adjacent players generally cannot be compared pairwise because of their prevalence during non-overlapping time periods.
\end{remark}


Next, we use \eqref{eq: W(s_ij | pi)} to compactly define a row-stochastic matrix $P \in \mathbb{R}_{\geq 0}^{n! \times n!}$ such that, for any given $\pi \in \mathfrak{S}_n$, we have the transition probabilities
\begin{equation}
\label{eq: P}
    P_{\pi,\sigma} 
    =
    \begin{cases}
        \frac{1}{2} \cdot \mathcal{W}(s_{ij} | \pi), & \text{ if } \sigma = s_{ij} \cdot \pi \text{ for some } 1 \leq i < j \leq n \\
        1/2, & \text{ if } \sigma = \pi \\
        0, & \text{ otherwise. } \\
    \end{cases}
\end{equation}
The matrix $P$ governs a data-driven random walk on $\mathfrak{S}_n$.
Note that its support along any given row is of size $O(n^2)$ and can be evaluated efficiently on the fly.

\begin{lemma}
\label{lemma: P}
Given~\eqref{eq: P}, $P$ is irreducible and aperiodic.
\end{lemma}
\begin{proof}
The additive $1$ term in \eqref{eq: W} implies the support graph of $P$ is the bi-directed Bruhat graph, with the addition of self-loops.
Fact~\ref{fact: product of 2-cycles} implies this graph is strongly connected, so that $P$ is irreducible.   
Moreover, the non-zero probability of a lazy step in the second case of \eqref{eq: W(s_ij | pi)} implies $P_{\pi,\pi} > 0$ for all $\pi \in \mathfrak{S}_n$, so that $P$ is aperiodic.
\end{proof}
Lemma~\ref{lemma: P} and Theorem~\ref{theorem: stationary} directly imply the following.
\begin{corollary}
\label{corollary: stationary}
$P$ has a unique stationary distribution $\mu$ over $\mathfrak{S}_n$.
\end{corollary}

\subsection{Partial Ordering}
\label{sec: partial ordering}

In light of Corollary~\ref{corollary: stationary}, consider the stationary distribution $\mu$ of $P$.
We combine this distribution with a suitable notion of stochastic dominance to obtain a partial order $([n], \preceq_\mu)$ over the players.
In particular, for each $i, j \in [n]$, let
\begin{equation}
\label{eq: dominance}
    j \preceq_\mu i \iff \Pr_{X \sim \mu}[X(i) \leq k] \geq \Pr_{X \sim \mu}[X(j) \leq k], \ \forall 1 \leq k \leq n.
\end{equation}
Recall that $\mu$ is a distribution over $\mathfrak{S}_n$, so that $X$ is a random permutation (in explicit notation) and $X(i) \in [n]$ denotes the ranking of player $i \in [n]$.
Therefore, at an intuitive level, \eqref{eq: dominance} states that player $i$ is conclusively ``better'' than player $j$ in $\preceq_\mu$ if and only if, given that $X \sim \mu$, player $i$ is across the board more likely to rank better than player $j$.

Let $\mathcal{L}(\preceq_\mu)$ denote the set of linear extensions of $\preceq_\mu$.
In the experimental part of this work, we will consider the average rank of a player with respect to a random linear extension of $\preceq_\mu$ as a summary statistic of the player's historical performance.
In particular, for each $i \in [n]$ we will consider
\begin{equation}
\label{eq: avg}
    \overline{X}(i) = \frac{1}{N} \sum_{t=1}^T X_t(i),
\end{equation}
for some $T \in \mathbb{N}$ and where $X_1, X_2, \ldots, X_T$ are sampled uniformly at random from $\mathcal{L}(\preceq_\mu)$.

\subsection{Sampling}
\label{sec: sampling}

Our next goal is to obtain independent samples from the stationary distribution $\mu$ induced by \eqref{eq: P} to form an empirical poset according to \eqref{eq: dominance}.
Based on Theorem~\ref{theorem: convergence}, this can be achieved through the repeated execution of \eqref{eq: P}.
At this point, the question of speed of convergence comes to the forefront.
If we knew the mixing time $\tau \in \mathbb{N}$ of this Markov chain for a desired tolerance to error, we could execute the random walk collecting every $\tau$th sample.
However, based on the discussion in Section~\ref{sec: related work}, we anticipate this is a very challenging theoretical question.
Neither is it easy to obtain numerical upper bounds based on the eigenvalues of the transition matrix $P$ due to its $n! \times n!$ size.

Given the applied scope of this work, we do not attempt to answer these questions and simply acknowledge that analytical questions around data-driven random walks on the symmetric group are both well-motivated and of independent mathematical interest.
In the experimental part of this work, we execute this random walk for a large number of steps collecting every $2(n-1)$th sample.
Given the $n-1$ diameter noted in Remark~\ref{remark: g_n} and the $1/2$ probability of taking a lazy step in \eqref{eq: P}, this is the expected number of steps needed for the entire state space to be reachable from any given state.

\section{Numerical Implementation}
\label{sec: numerical implementation}

In this section, we implement our methods using singles rankings from the ``Open Era'' of the ATP and the WTA.
As discussed in Section~\ref{sec: cutoffs and the meaning of GOAT}, we experiment with taking the top $3$, $5$, $10$, and $20$ players as different cutoffs from the full ranking data that is available.
For each experiment, we execute a random walk as described in Section~\ref{sec: sampling} until we collect $100,000$ samples.
We then use these samples to obtain data-driven posets over the players, as described in Section~\ref{sec: partial ordering}.
We illustrate these posets as horizontal Hasse diagrams with maximal elements on the right-hand side.
The full Hasse diagrams we obtain are shown in Figure~\ref{fig: full hasse}, but in short we will zoom in on their top portions.
\begin{figure}[ht]
    \begin{subfigure}{0.495\linewidth}
        \includegraphics[width=\linewidth]{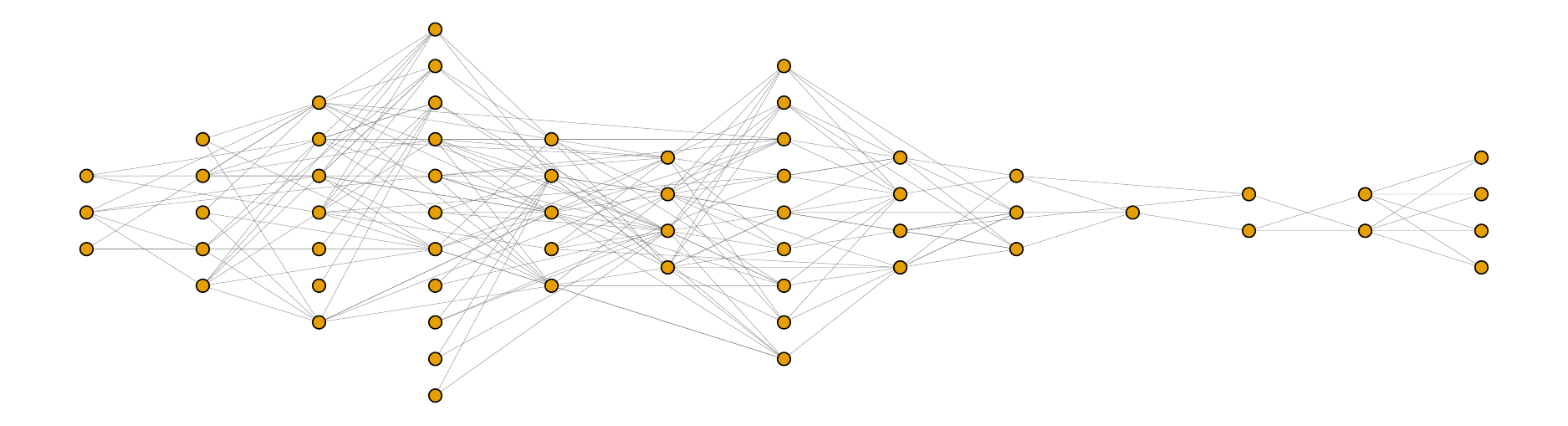}
        \includegraphics[width=\linewidth]{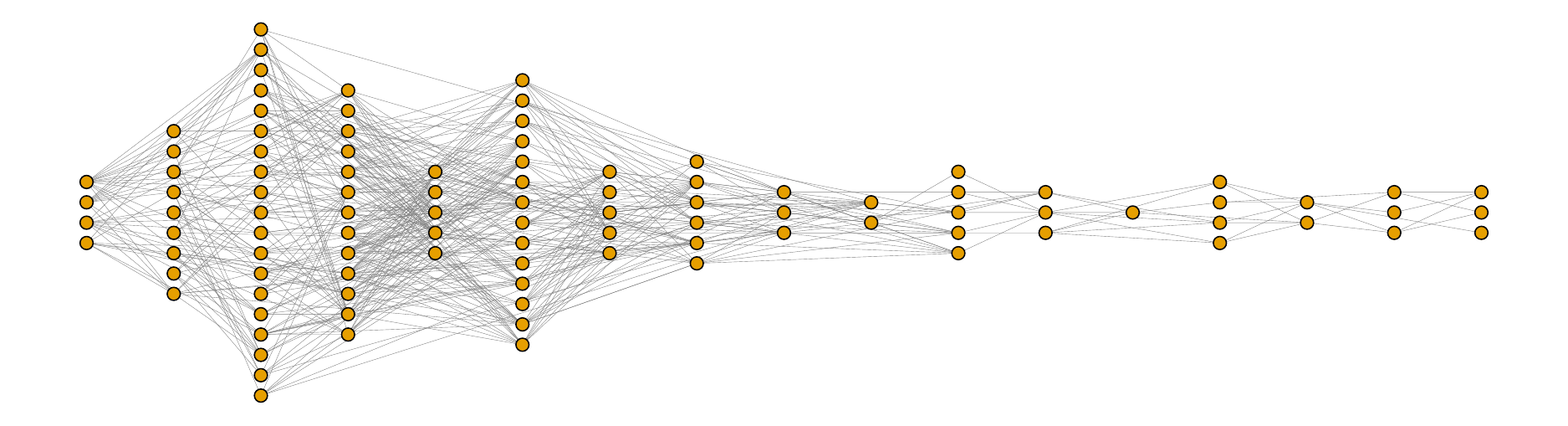}
        \includegraphics[width=\linewidth]{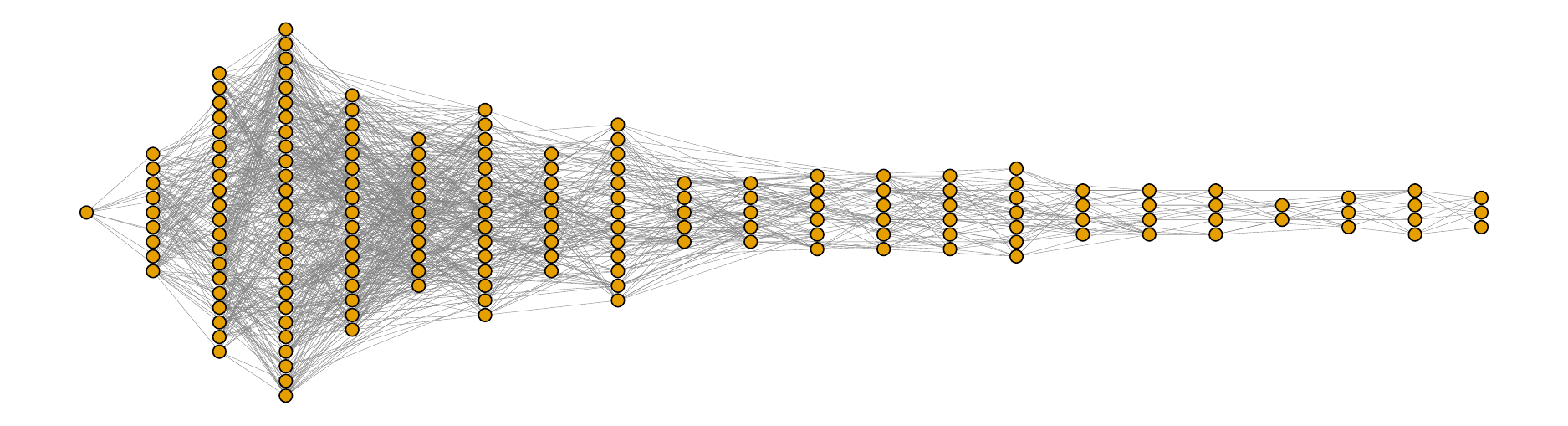}
        \includegraphics[width=\linewidth]{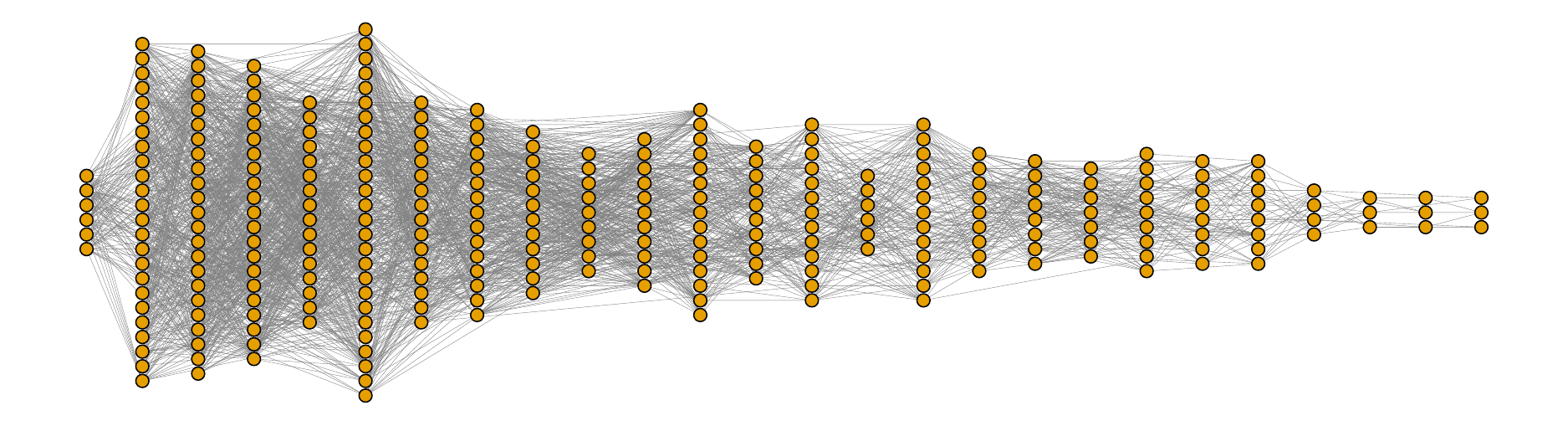}
    \caption{
        Rankings from the ATP.
    }
    \centering
    \end{subfigure}
    \begin{subfigure}{0.495\linewidth}
        \includegraphics[width=\linewidth]{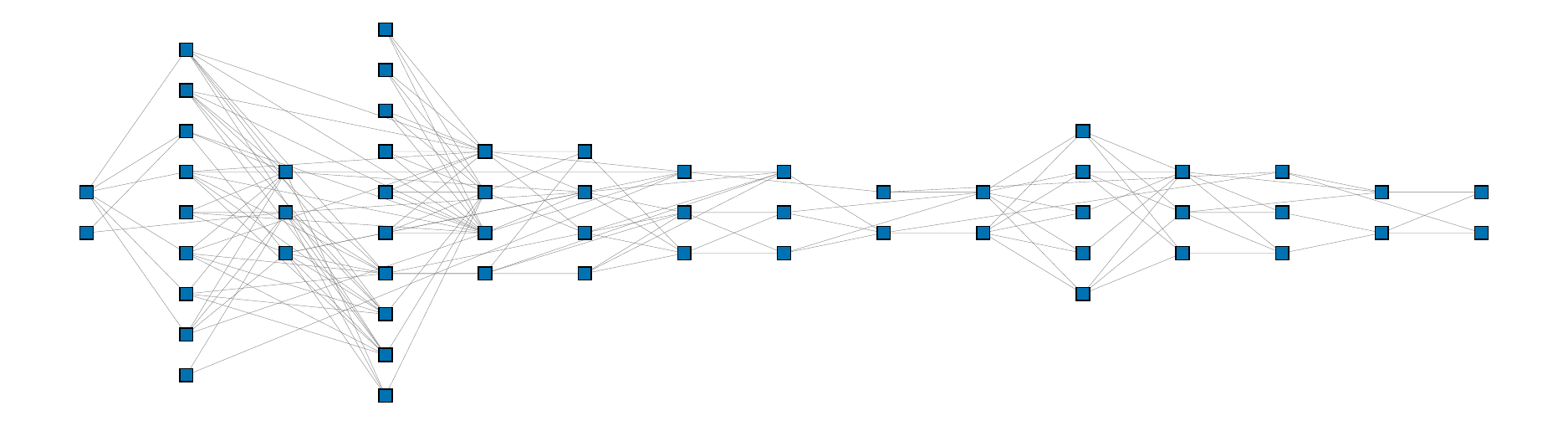}
        \includegraphics[width=\linewidth]{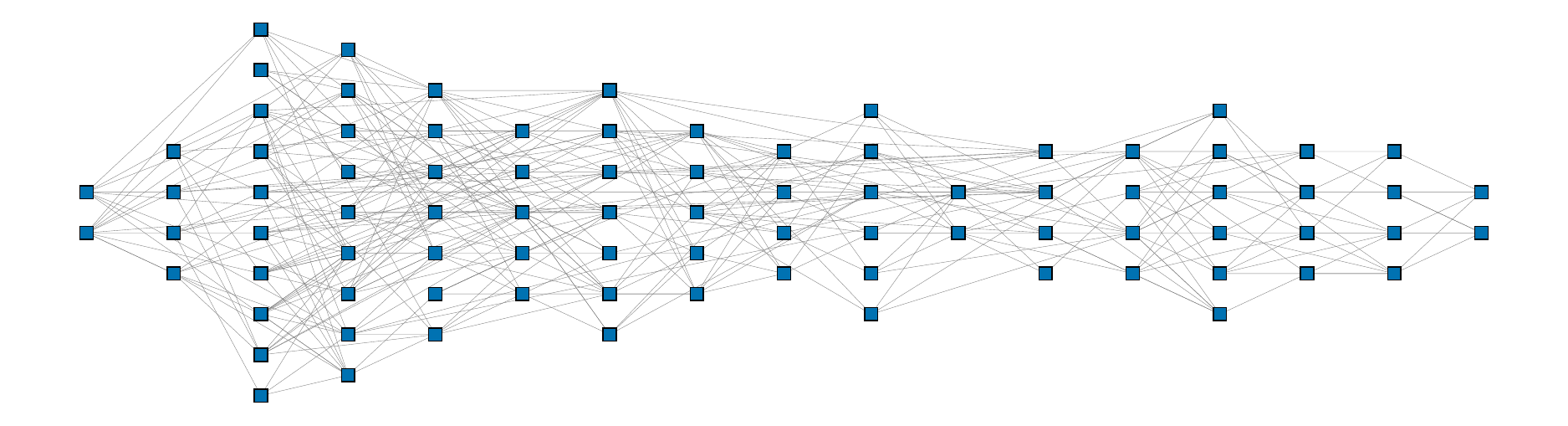}
        \includegraphics[width=\linewidth]{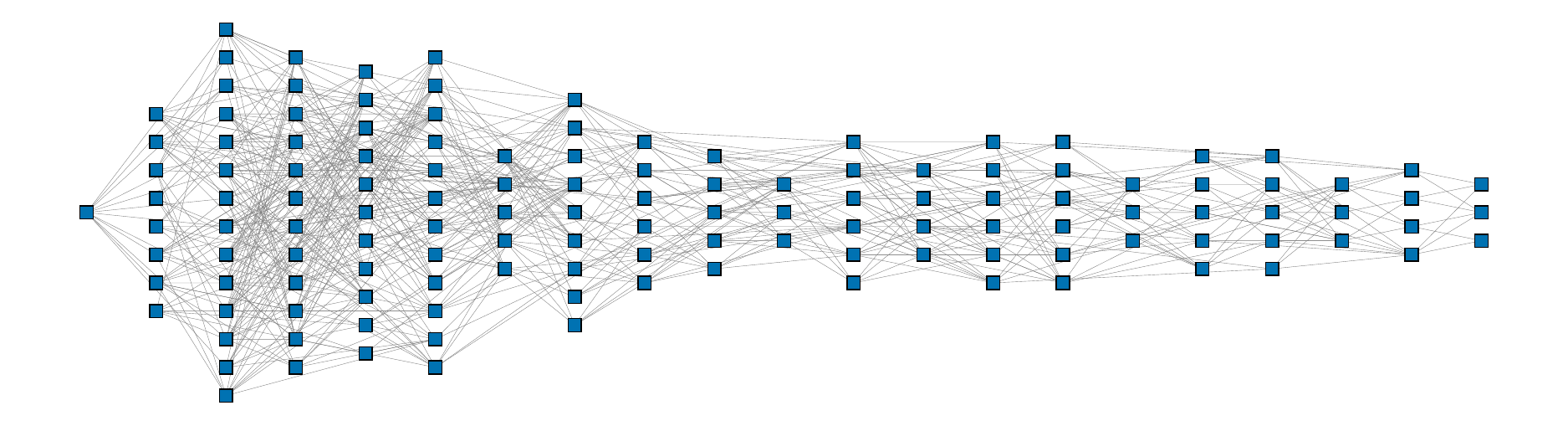}
        \includegraphics[width=\linewidth]{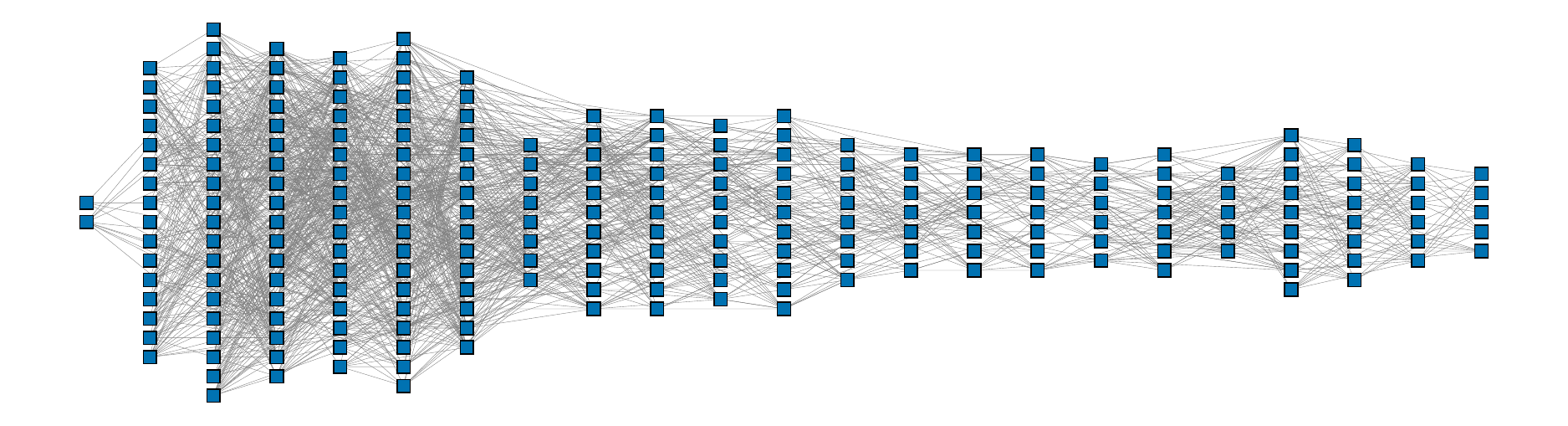}
    \caption{
        Rankings from the WTA.
    }
    \centering
    \end{subfigure}
    \caption{
    Full Hasse diagrams from the experiments in Section~\ref{sec: numerical implementation}.
    From top to bottom, the rows correspond to top the posets over top 3, 5, 10, and 20 players, respectively.
    }
    \label{fig: full hasse}
\end{figure}
To estimate the average rank of the players in a given poset, we first sample $100,000$ linear extensions uniformly at random using the package \texttt{SageMath} by~\cite{sagemath} and then follow \eqref{eq: avg}.

The remainder of this section is organized as follows.
In Section~\ref{sec: ranking data} we further describe the datasets we use for our experiments.
In Section~\ref{sec: rankings from the atp} we discuss our results using rankings from the ATP, whereas in Section~\ref{sec: rankings from the wta} we discuss our results using rankings from the WTA.
Refer to Tables~\ref{tab: atp} and~\ref{tab: wta} for a summary of some of our findings.
Lastly, in Section~\ref{sec: comparing the associations} we compare the results obtained for the ATP and the WTA to note qualitative differences between the two associations.
All mentioned statistics that are external to our numerical implementation can be verified online~\cite{ATP_Tour_2024, WTA_Tennis_2024}. 
The source code for these experiments is available at \texttt{github.com/jcmartinezmori/goat}.

\subsection{Ranking Data}
\label{sec: ranking data}

Both the ATP and the WTA singles rankings are based on rolling $52$-week period systems that award points to individual players according to their placement in different tournament categories.
Rankings are typically published on Mondays, and points are dropped $52$ weeks after being awarded.
Note that the point systems have changed over the years: details on the current systems are available online in the official rulebooks~\cite{atp_rulebook_2024, wta_rulebook_2024}.

We retrieved historical rankings from the repositories compiled by Sackmann~\cite{tennisATP,tennisWTA} for both the ATP and the WTA.
In the case of the ATP, this corresponds to $2230$ weekly rankings from $1973$ to $2023$ (Sackmann notes that this data is complete from 1985, but that it is intermittent from $1973$ to $1984$ and that it is missing $1982$). 
In the case of the WTA, this corresponds to $2052$ weekly rankings from $1984$ to $2023$.

Naturally, the limitations of the data we use induce limitations on our experimental results.
Aside from missing data points, we in particular do not consider the ``Amateur Era'' and only partially capture players that were active during the the onset of the ``Open Era'' rankings.
As a concrete example, we do not consider Margaret Court; a widely successful player in women's professional tennis that was active from 1960 to 1977.
This data availability issue speaks to the difficulty of capturing the Amateur Era due to the lack of a centralized ranking system.

\subsection{Rankings from the ATP}
\label{sec: rankings from the atp}

In Figure~\ref{fig: atp-3} we plot the top portion of the Hasse diagram of the poset over ``top 3 players,''  where Djokovic, Federer, Nadal, and Sampras are incomparable maximal elements\textemdash no one of these players is conclusively better than another.
Note also that they all have a similar average rank.
Conversely, we conclude that each of Djokovic, Federer, Nadal, and Sampras is better than Connors and Lendl as a top 3 player.
Similarly, as a top 3 player, Connors is better than Andre Agassi, who is better than Bj\"orn Borg, who is better than Andy Murray, who is better than Carlos Alcaraz.
In other words, these players form a chain in the poset.
Formally, this is
\begin{equation*}
    \text{Djokovic} \succ_\mu \text{Connors} \succ_\mu \text{Agassi} \succ_\mu \text{Borg} \succ_\mu \text{Murray} \succ_\mu \text{Alcaraz}.
\end{equation*}
This is also reflected by significant jumps in the average rank from one player to the next.
Interestingly, Alcaraz is near the top of the Hasse diagram despite turning professional most recently, in 2018.
\begin{figure}[ht]
    \centering
    \includegraphics[width=\linewidth]{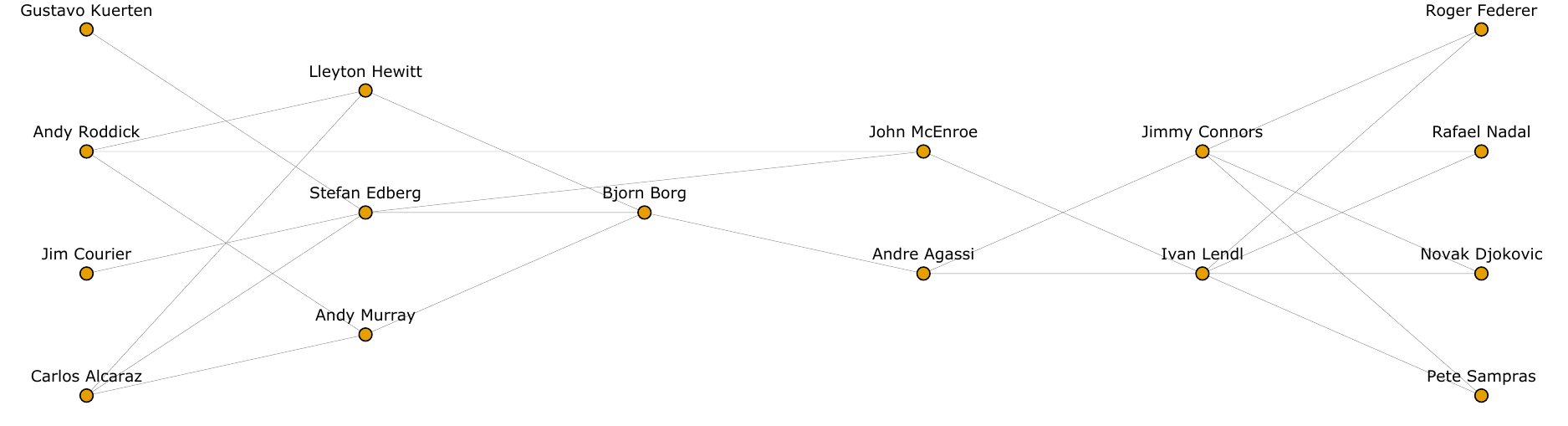}
    \includegraphics[width=\linewidth]{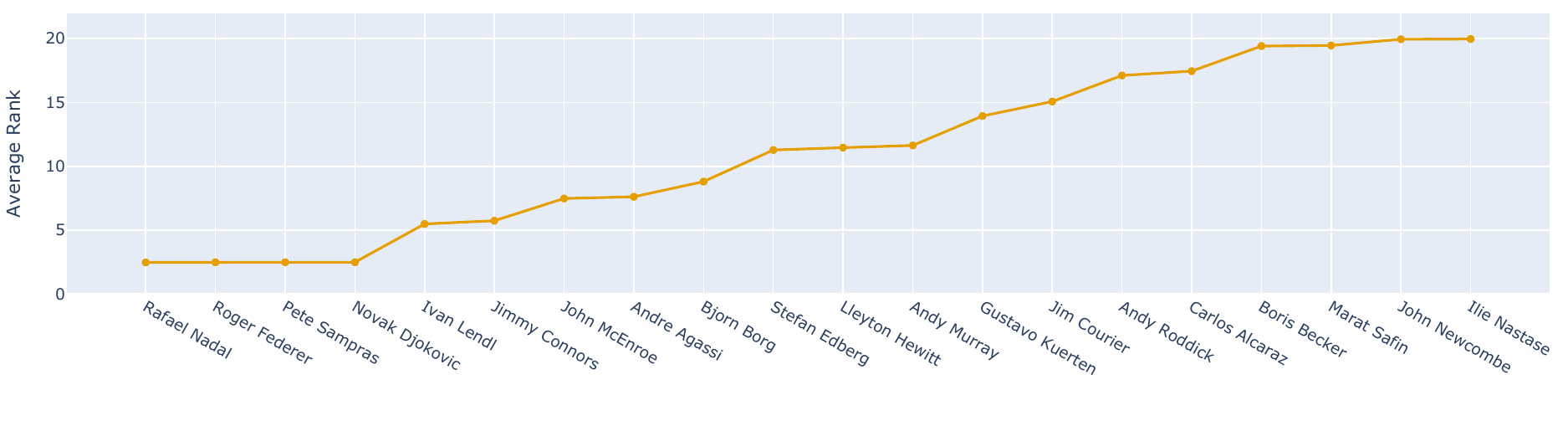}
    \caption{
        Cutoff at top $3$ of full ATP rankings.
    }
    \label{fig: atp-3}
\end{figure}

In Figure~\ref{fig: atp-5} we plot the top part of the Hasse diagram of the poset over ``top 5 players,'' where Djokovic, Federer and Nadal remain incomparable maximal elements, although Federer fares slightly better than the other two in terms of average rank.
However, Sampras is no longer a maximal element: as a top 5 player, he is worse than both Federer and Nadal, and he is incomparable to Djokovic.
Sampras is moreover incomparable to Connors and Lendl (both of whom he was better than as a top 3 player).
Note that Connors, Djokovic, and Lendl are also incomparable as top 5 players: their difference in average rank favors Djokovic, but it is not sufficiently conclusive.
Moreover, note that Alcaraz no longer appears near the top of the Hasse diagram, having been displaced by other players who are better top 5 players than they are top 3 players, such as Guillermo Vilas (this is reflected in the average rank plot).
First, note that Alcaraz and Vilas did not overlap in competition: Vilas played professionally from 1968 to 1992.
Second, note that Alcaraz transitioned very quickly from below the top 20 to within the top 3 from February to September of 2022, and has consistently prevailed within the top 3 until the time of writing.
His prevalence within the top 3 is significant independently of his short career at present.
However, when the cutoff is relaxed to top 5, his short career-to-date puts him at a disadvantage compared to retired players with significant career-long prevalence in the top 5 (but perhaps not in the top 3), such as Vilas.
\begin{figure}[ht]
    \centering
    \includegraphics[width=\linewidth]{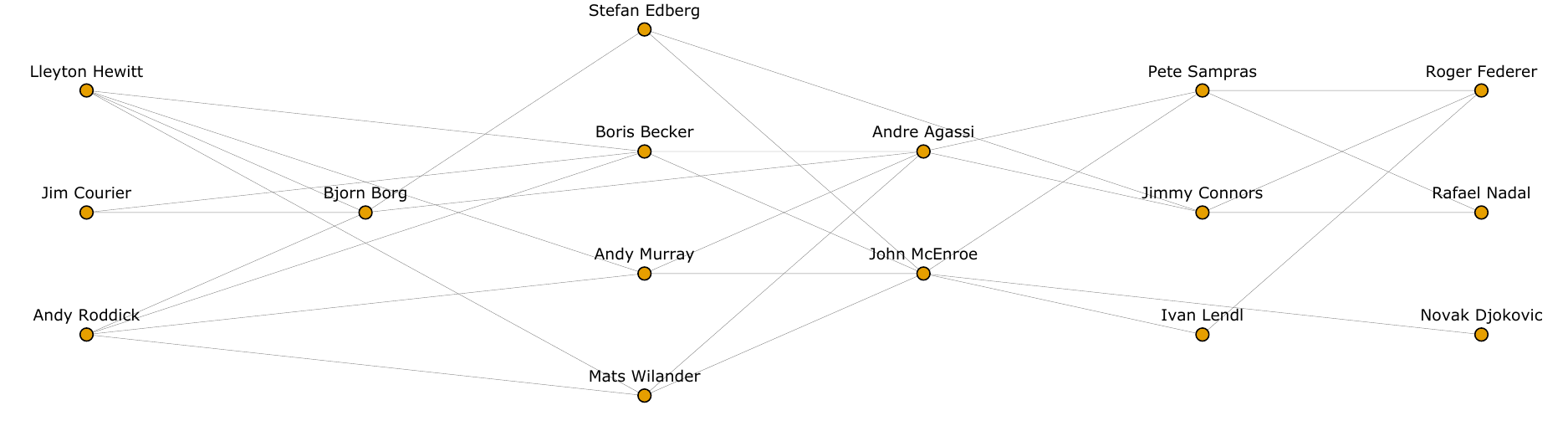}
    \includegraphics[width=\linewidth]{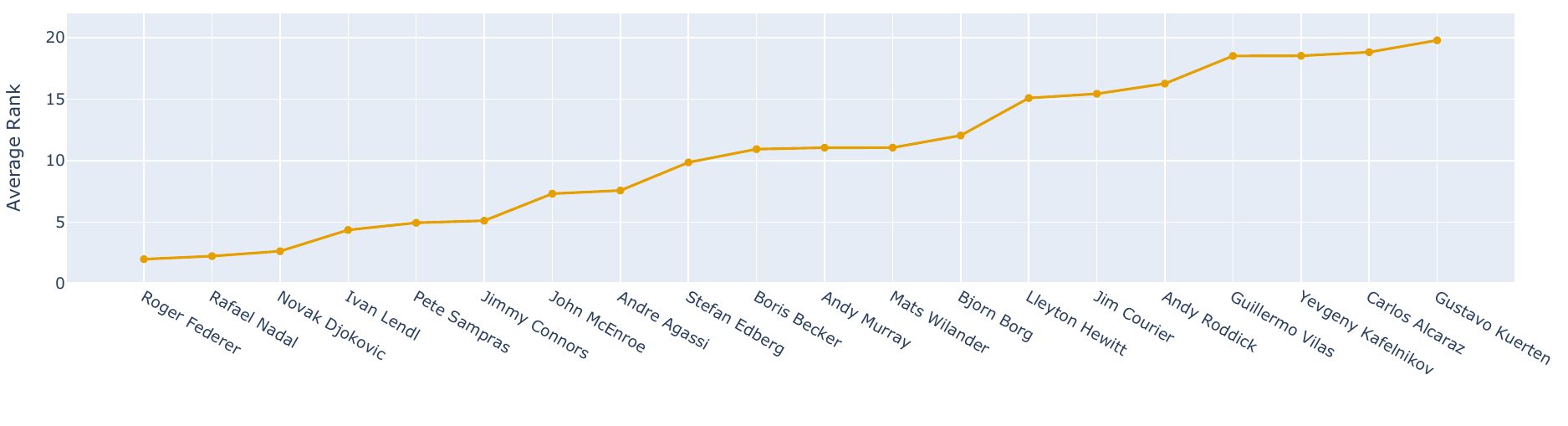}
    \caption{
        Cutoff at top $5$ of full ATP rankings.
    }
    \label{fig: atp-5}
\end{figure}

In Figure~\ref{fig: atp-10} we plot the top part of the Hasse diagram of the poset over ``top 10 players,'' where Djokovic, Federer, and Nadal once again remain incomparable maximal elements.
However, unlike the previous cutoff, we conclude that Djokovic is better than Connors and Sampras as a top 10 player.
Qualitatively, we observe the inclusion of more players near the top of this Hasse diagram compared to the previous cases, indicating that there are more ``good top 10 players'' than there are good top 3 or top 5 players.
For example, Becker (who was far from maximal in Figures~\ref{fig: atp-3} and ~\ref{fig: atp-5}) is now close to maximal: only Federer and Nadal are better than him, and he is incomparable to Djokovic.
\begin{figure}[ht]
    \centering
    \includegraphics[width=\linewidth]{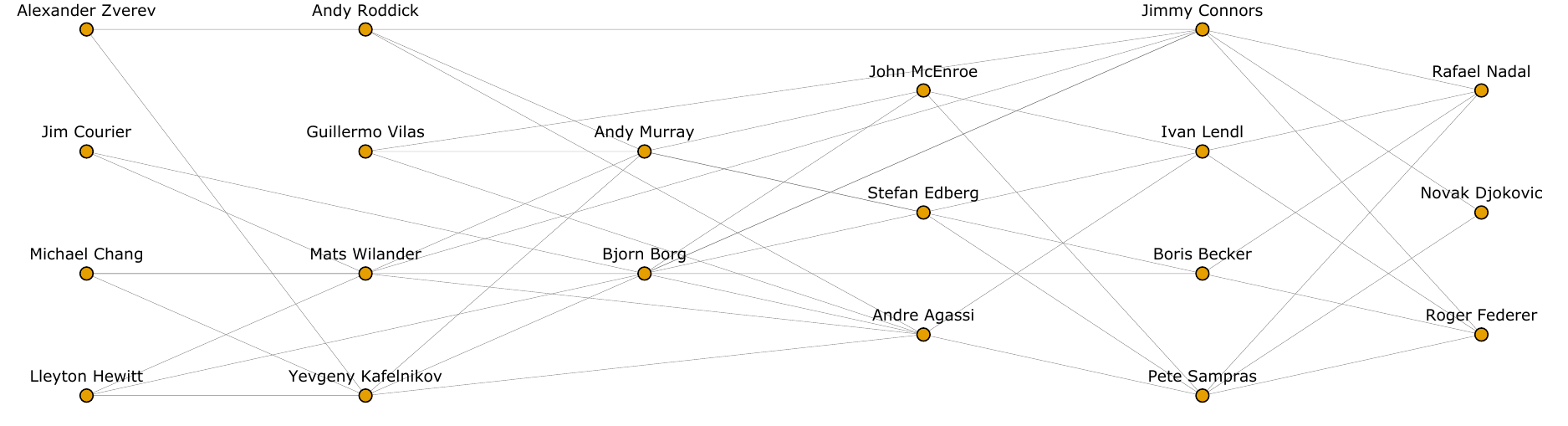}
    \includegraphics[width=\linewidth]{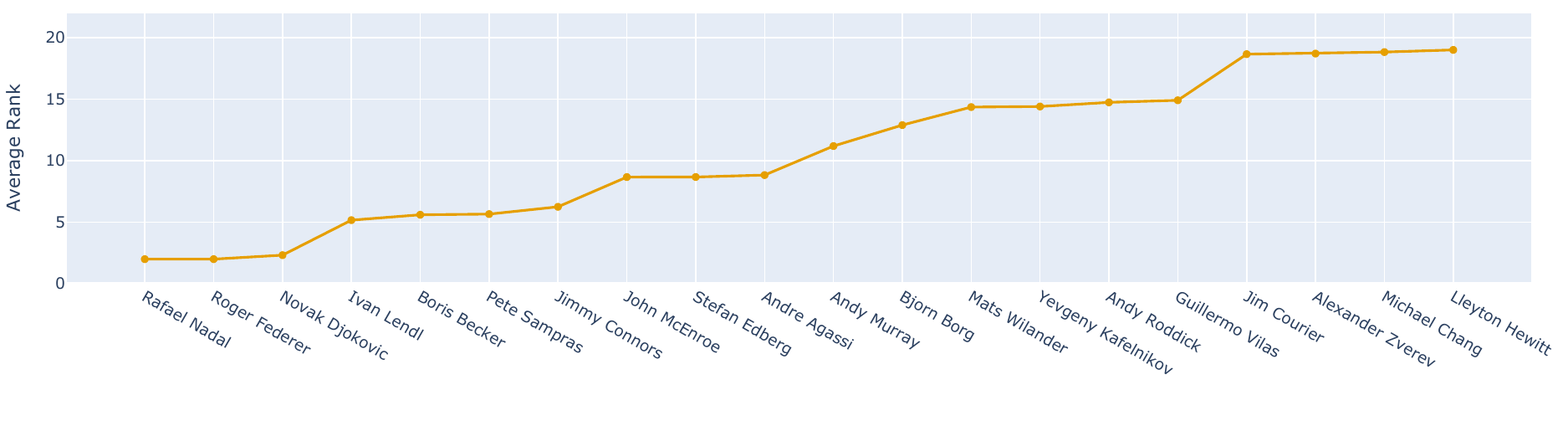}
    \caption{
        Cutoff at top $10$ of full ATP rankings.
    }
    \label{fig: atp-10}
\end{figure}

Lastly, in Figure~\ref{fig: atp-20} we plot the top part of the Hasse diagram of the poset over ``top 20 players.''
Interestingly, Federer and Nadal are joined by Lendl as incomparable maximal elements.
All three of them are better than Djokovic as a top 20 player.
Qualitatively, we again observe the inclusion of even more players near the top of this Hasse diagram compared to the previous cases.
For instance, the Djokovic $\succ_\mu$ Connors $\succ_\mu$ Agassi chain from Figure~\ref{fig: atp-3} is broken in Figure~\ref{fig: atp-20}, where all three of these players are equally close to the top of the Hasse diagram and form an anti-chain.
\begin{figure}[ht]
    \centering
    \includegraphics[width=\linewidth]{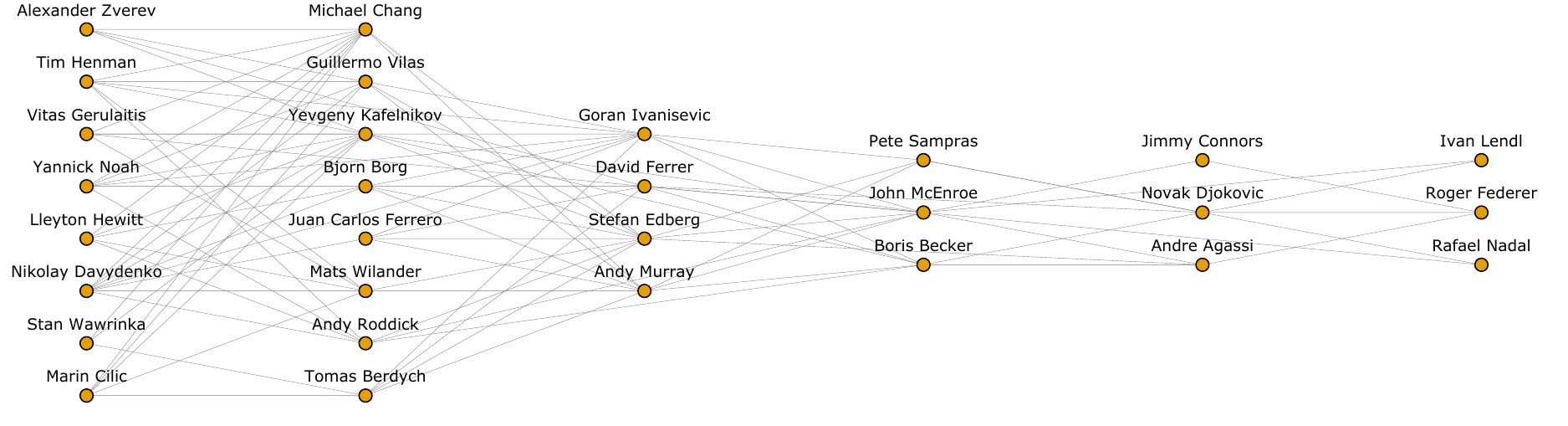}
    \includegraphics[width=\linewidth]{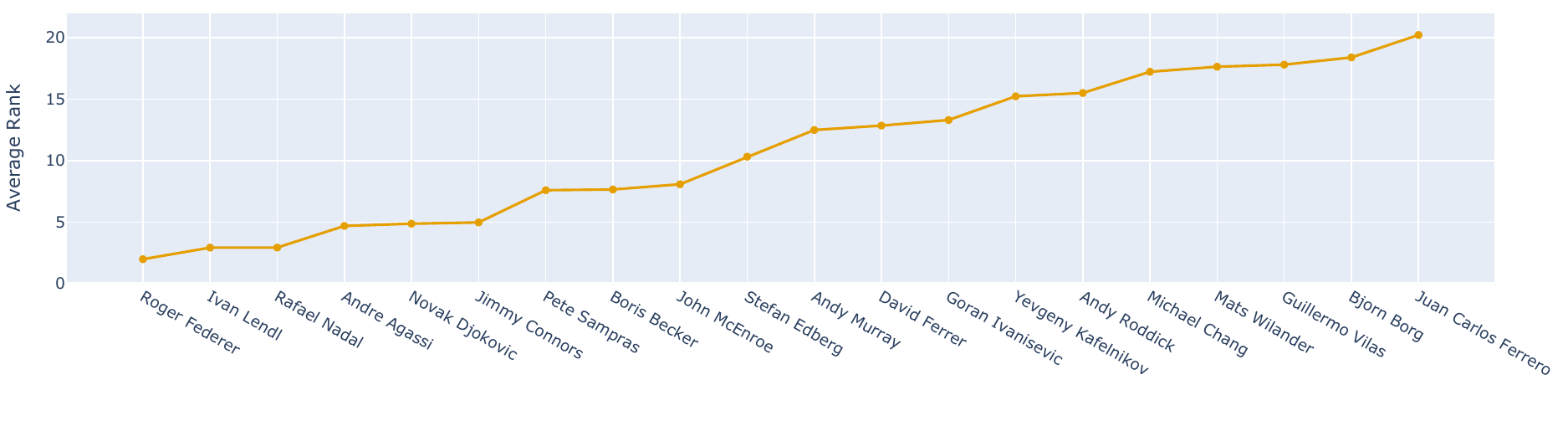}
    \caption{
        Cutoff at top $20$ of full ATP rankings.
    }
    \label{fig: atp-20}
\end{figure}

\subsection{Rankings from the WTA}
\label{sec: rankings from the wta}

In Figure~\ref{fig: wta-3} we plot the top portion of the Hasse diagram of the poset over ``top 3 players,''  where Graf and S. Williams are incomparable maximal elements\textemdash no one of these players is conclusively better than the other.
Note also that they both have a similar average rank, albeit with a slight advantage for Graf.
Note that S. Williams and Martina Navratilova are incomparable.
On the other hand, we conclude that, as top 3 players, Graf is better than Navratilova, who is better than Justine Henin, who is better than Aryna Sabalenka.
In other words, these players form a chain in the poset.
Formally, this is
\begin{equation*}
    \text{Graf} \succ_\mu \text{Navratilova} \succ_\mu \text{Henin} \succ_\mu \text{Sabalenka}.
\end{equation*}
We also conclude that, as a top 3 player, Sabalenka, Kim Clijsters, and Maria Sharapova are pairwise incomparable (i.e., they form an anti-chain).
\begin{figure}[ht]
    \centering
    \includegraphics[width=\linewidth]{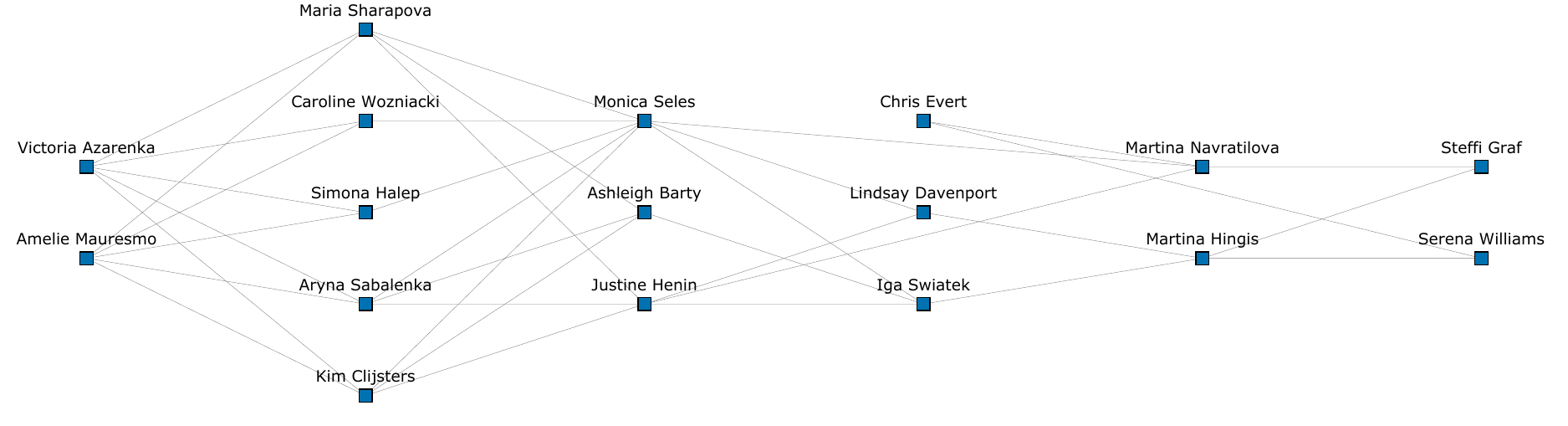}
    \includegraphics[width=\linewidth]{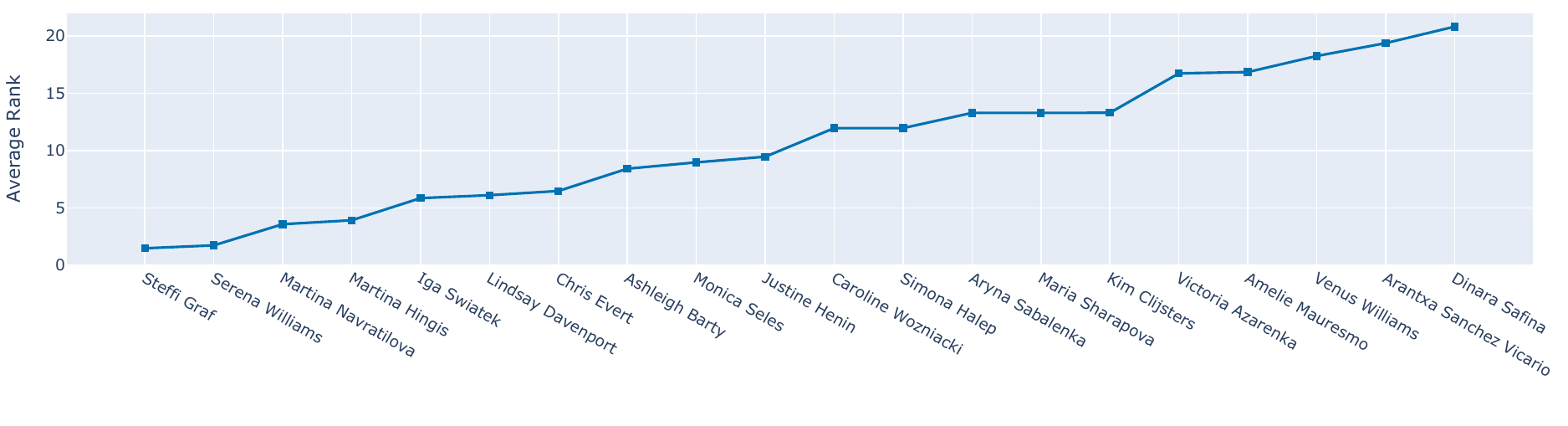}
    \caption{
        Cutoff at top $3$ of full WTA rankings.
    }
    \label{fig: wta-3}
\end{figure}

\begin{figure}[ht]
    \centering
    \includegraphics[width=\linewidth]{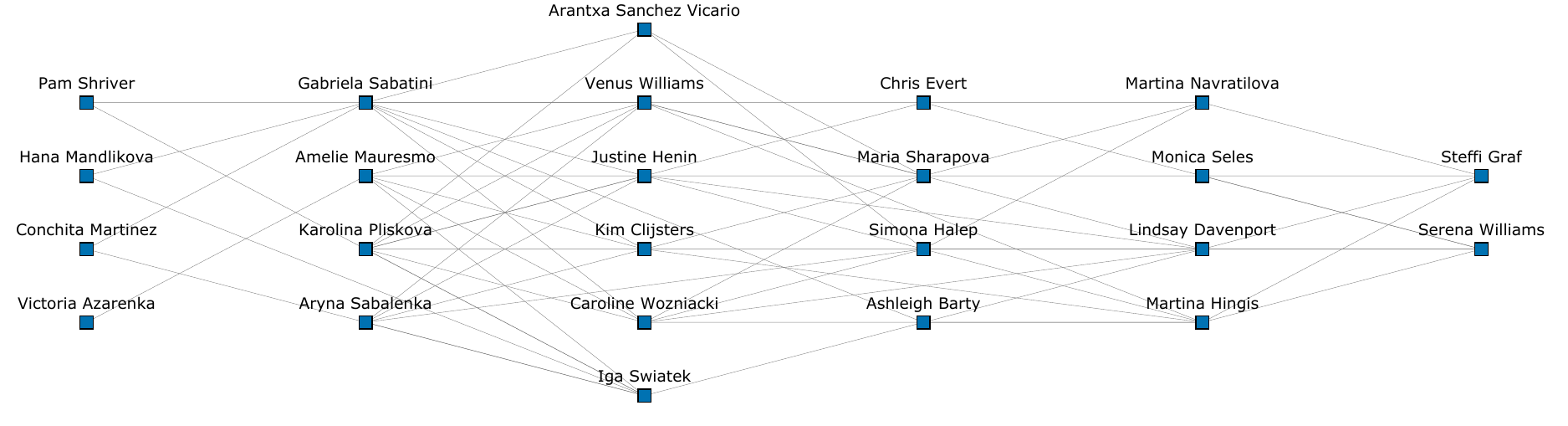}
    \includegraphics[width=\linewidth]{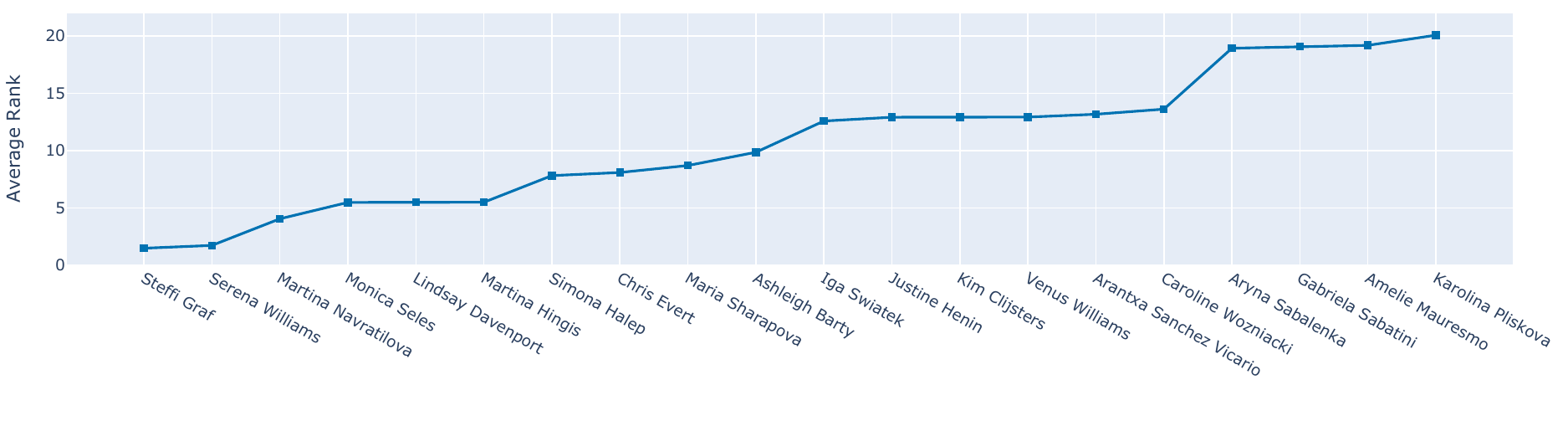}
    \caption{
        Cutoff at top $5$ of full WTA rankings.
    }
    \label{fig: wta-5}
\end{figure}
In Figure~\ref{fig: wta-5} we plot the top portion of the Hasse diagram of the poset over ``top 5 players,''  where Graf and S. Williams are again incomparable maximal elements.
Note that while Graf is better than Navratilova, S. Williams and Navratilova are incomparable.
Moreover, note that unlike the case in Figure~\ref{fig: wta-3}, as top 5 players, both Clijsters and Sharapova are better than Sabalenka.
This is a similar effect to that seen with Alcaraz in Section~\ref{sec: rankings from the atp}: Sabalenka turned professional in 2015 and had little overlap in competition with Sharapova, who retired in 2020, or with Clijsters, who retired in 2012 except for a brief comeback between 2020 and 2022.
In particular, her prevalence within the top 3 is significant independently of her short career at the time of writing.
However, when the cutoff is relaxed to top 5, her short career-to-date puts her at a disadvantage compared to long-established players with significant career prevalence in the top 5, such as Clijsters and Sharapova.
Qualitatively, we observe the inclusion of more players near the top of this Hasse diagram compared to the previous cases, and that chains of players that were comparable as top 3 players become anti-chains and close to maximal as top 5 players (e.g., Lindsay Davenport and Seles).

In Figure~\ref{fig: wta-10} we plot the top portion of the Hasse diagram of the poset over ``top 10 players,''  where now the three of Graf, Navratilova, and S. Williams are incomparable maximal elements.
In other words, while Navratilova is close to but not quite the GOAT as a top 3 or top 5 player, she is conclusively one of the GOATs as a top 10 player.
As a top 10 player, both Graf and Navratilova are better than Seles.
However, S. Williams cannot be conclusively compared to Seles even though she fares better than her in terms of average rank.
\begin{figure}[ht]
    \centering
    \includegraphics[width=\linewidth]{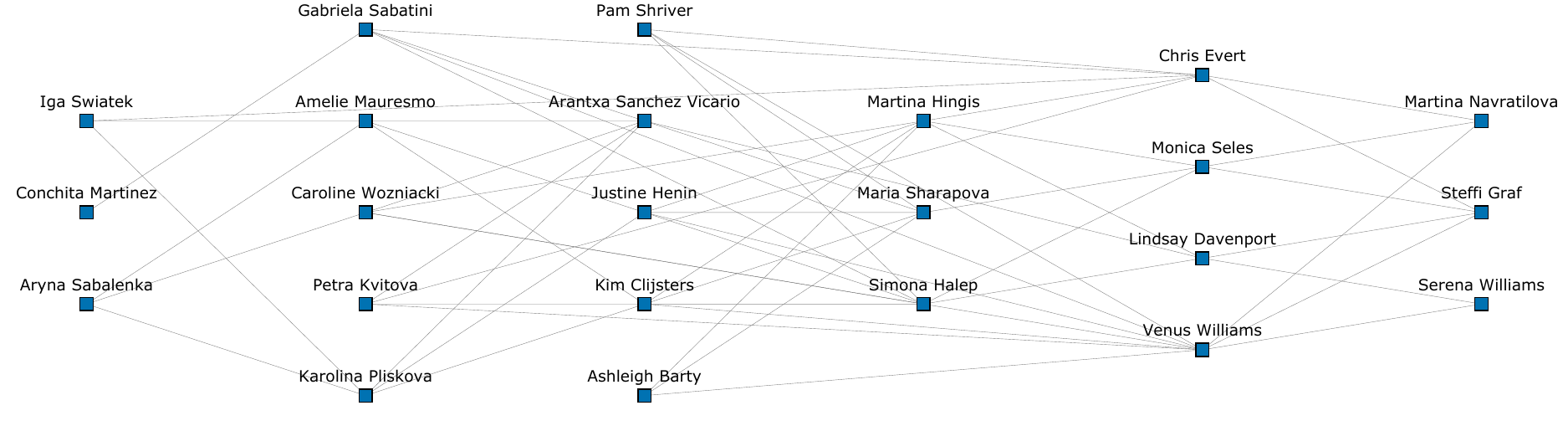}
    \includegraphics[width=\linewidth]{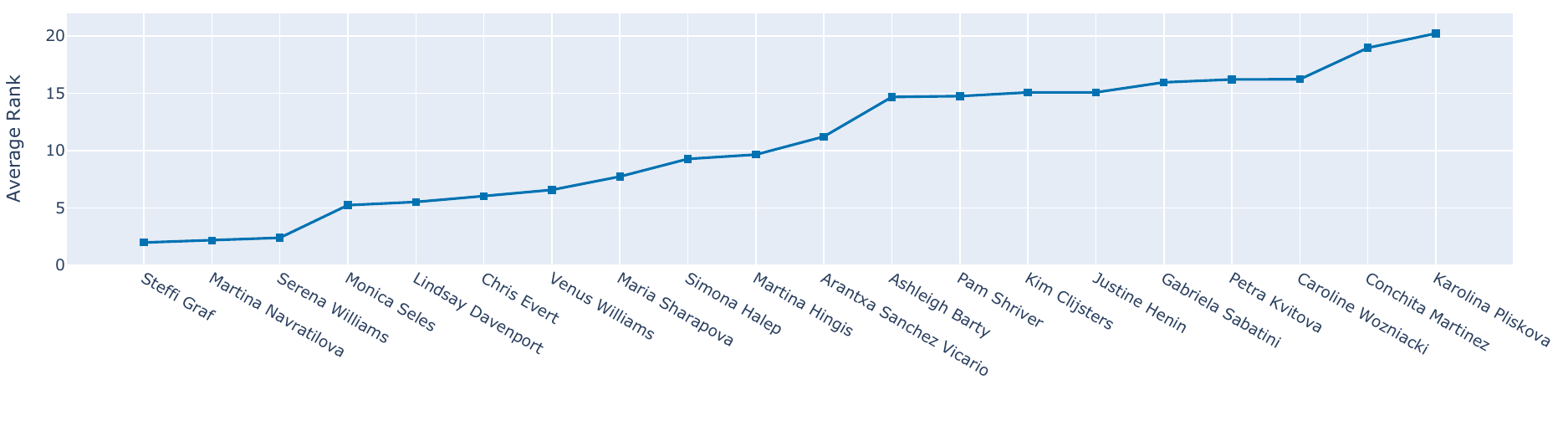}
    \caption{
        Cutoff at top $10$ of full WTA rankings.
    }
    \label{fig: wta-10}
\end{figure}

Lastly, in Figure~\ref{fig: wta-20} we plot the top portion of the Hasse diagram of the poset over ``top 20 players,''  where now all five of Graf, Navratilova, S. Williams, V. Williams, and Wozniacki are incomparable maximal elements.
Qualitatively, we observe a significant and across the board increment in the number of players near the top of the Hasse diagram.
This indicates that, in the WTA, there are many more good top 20 players than there are good top 3, 5, or 10 players.
\begin{figure}[ht]
    \centering
    \includegraphics[width=\linewidth]{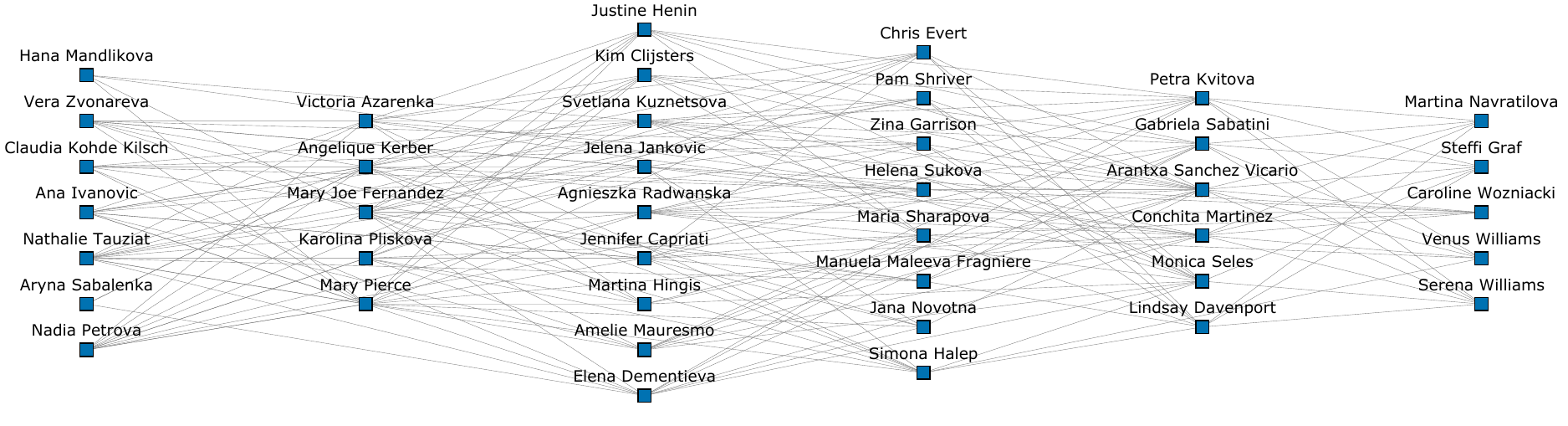}
    \includegraphics[width=\linewidth]{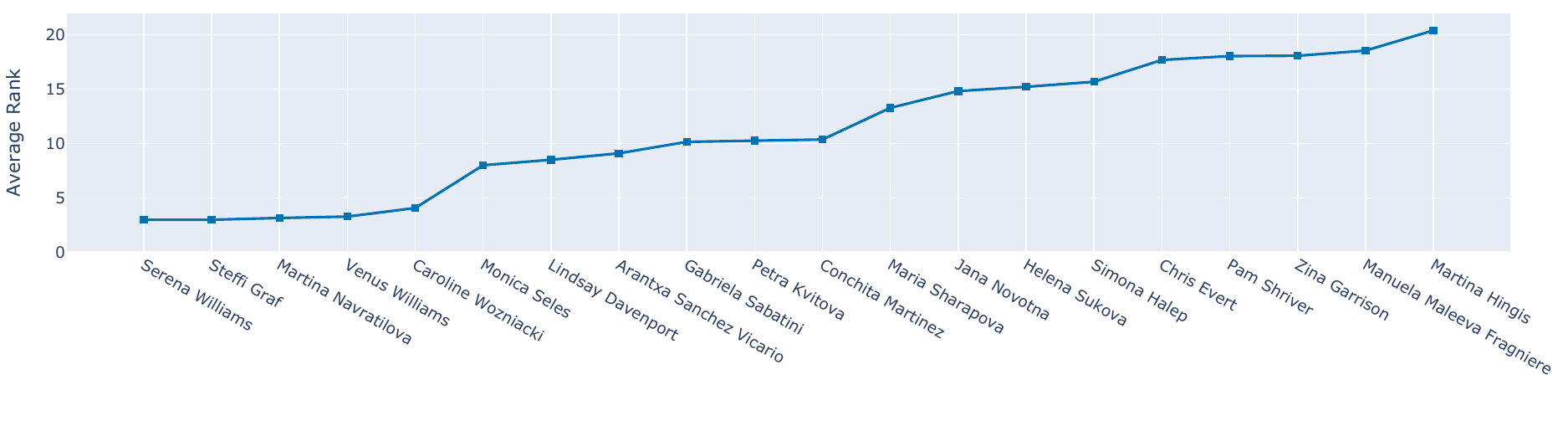}
    \caption{
        Cutoff at top $20$ of full WTA rankings.
    }
    \label{fig: wta-20}
\end{figure}

\subsection{Comparing the Associations}
\label{sec: comparing the associations}

Figures~\ref{fig: atp-20} and \ref{fig: wta-20} are qualitatively different in that, as top 20 players, the number of players close to maximal is much larger for the WTA than it is for the ATP.
Comparing the results from Sections~\ref{sec: rankings from the atp} and \ref{sec: rankings from the wta}, this difference is not nearly as apparent for the posets over top 3, 5, or 10 players. 
Here, we further investigate this difference from a more quantitative point of view.
In particular, we ask questions of the form: to what extent is being ``a good top $\kappa$ player'' a predictor of being ``a good top $\kappa'$ player,'' where $\kappa' \leq \kappa$?
We formalize the notion of ``a good top $\kappa$ player'' by average rank (from the linear extensions) for the poset cutoff $\kappa$.
Intuitively, the smaller the difference between $\kappa$ and $\kappa'$, the stronger we expect the correlation to be.

In Figure~\ref{fig: scatter}, we consider the set of players that appear among the first $50$ by average rank for all cutoffs $\kappa = 3, 5, 10, 20$.
This amounts to $35$ players in the WTA and $32$ the ATP.
Then, for each of these sets of players, we plot a matrix of scatter plots of average rank and varying cutoffs.
In this way, each point represents a player, and its coordinates correspond to the player's average rank for a combination of cutoffs.
For each scatter plot, we also fit a linear model and report its coefficient of determination $r^2$.
\begin{figure}[ht]
    \centering
    \begin{subfigure}{0.495\linewidth}
    \includegraphics[width=\linewidth]{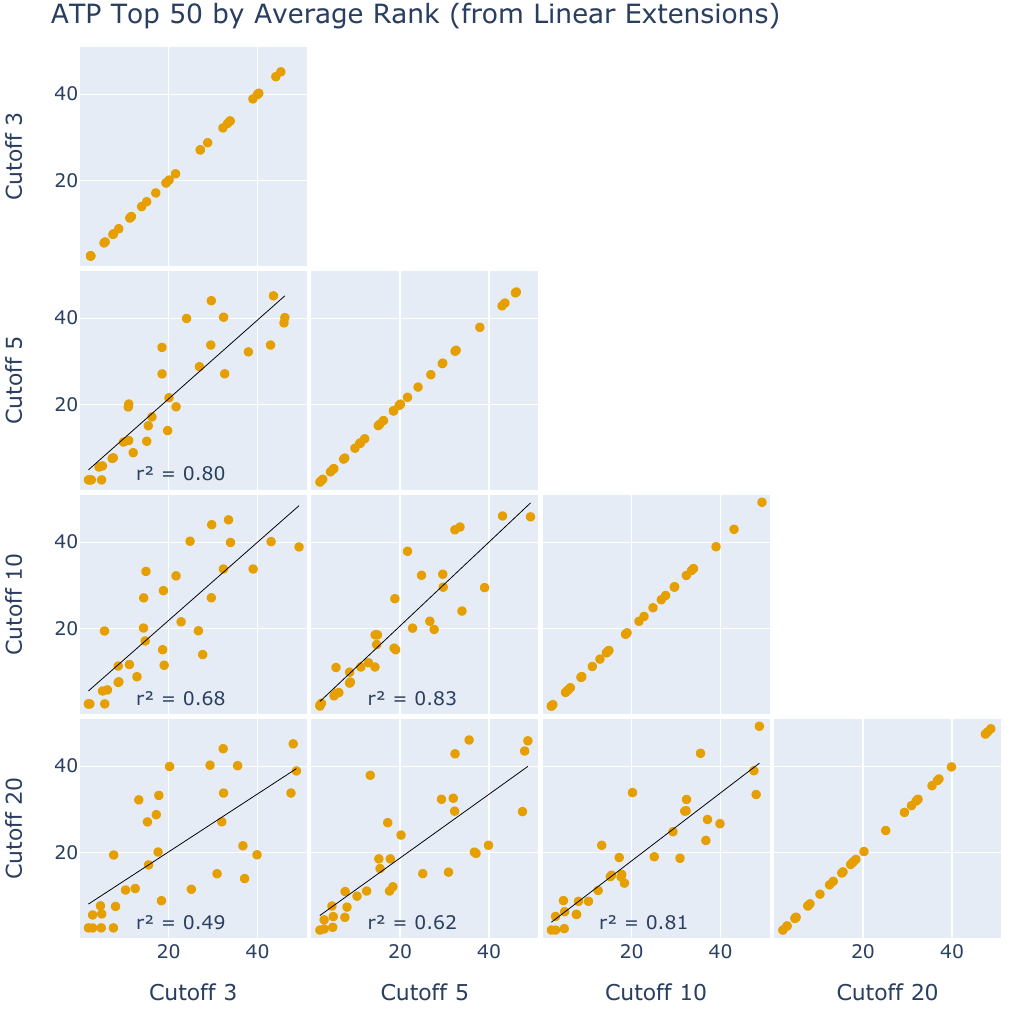}
    \subcaption{Rankings from the ATP.}
    \end{subfigure}
    \begin{subfigure}{0.495\linewidth}
    \includegraphics[width=\linewidth]{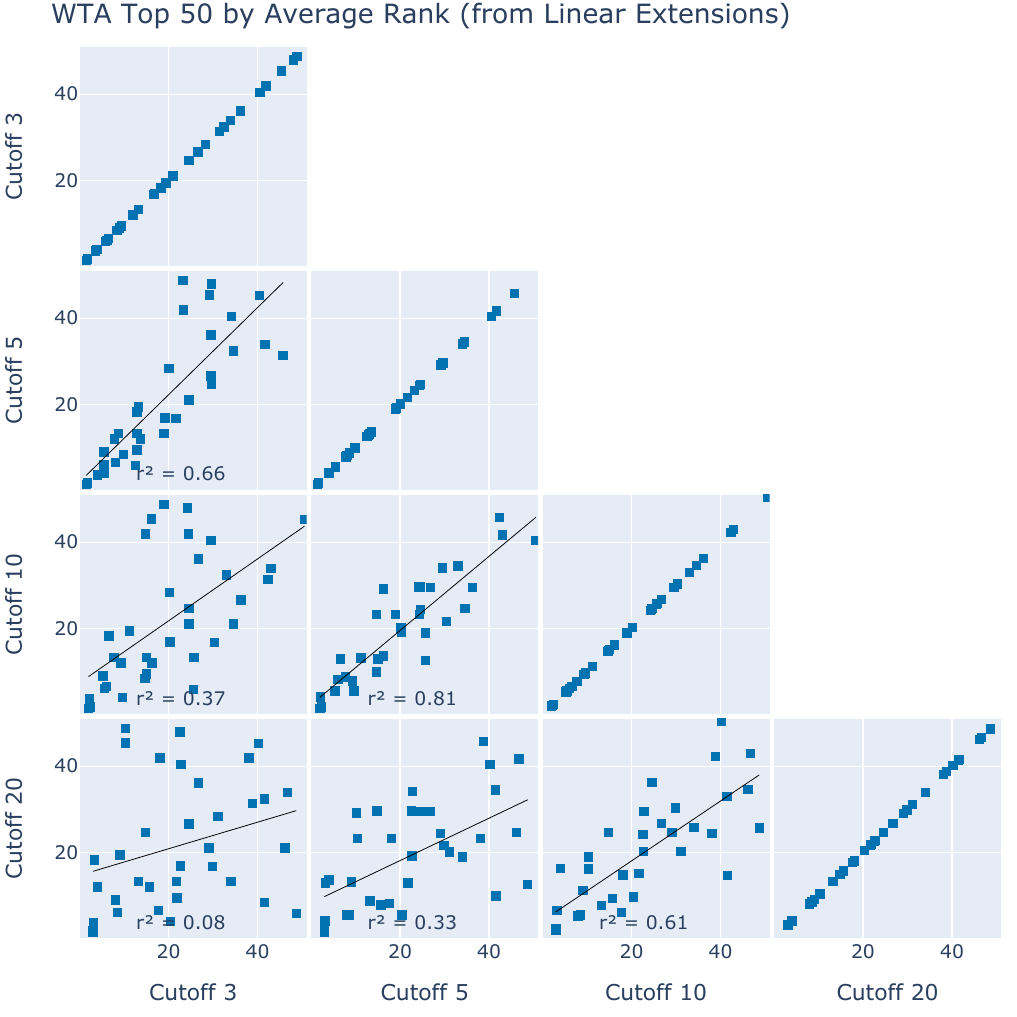}
    \subcaption{Rankings from the WTA.}
    \end{subfigure}
    \caption{Average ranks (from linear extensions) with varying cutoffs.}
    \label{fig: scatter}
\end{figure}

From Figure~\ref{fig: scatter} we note that, in both associations, being a good top 5 player is a reasonable predictor of being a good top 3 player, as is being a good top 10 player a reasonable predictor for being a good top 5 player.
In both cases, being a good top 10 player is still correlated with being a good top 3 player, although much less so than with being a good top 5 player.
However, when we consider being a good top 20 player as a predictor, we observe a major difference between the two associations.
In the ATP, being a good top 20 player is correlated with being a good top 10, 5, and 3 player (with decreasing correlation in that order).
Conversely, in the WTA, being a good top 20 player is only mildly correlated with being a good top 10 player, and it is nearly uncorrelated with being a good top 5 or top 3 player.
In general, all correlations are weaker in the WTA than they are in the ATP.
This quantitative finding highlights a major difference in terms of career and dominance longevity for top players across the two associations.
While we do not delve into causality in this work, we believe it would be practically important to further study the reasons and broader implications of this phenomenon.

\section{Conclusion}
\label{sec: conclusion}

In this work, we proposed a mathematical framework to compare players, and in particular players spanning across different time periods in sports history.
As illustrated in Figure~\ref{fig: summary}, we aggregated cutoffs of historical ranking data into a single poset over top players by \textit{i}) sampling from a data-driven random walk over the symmetric group and \textit{ii}) adopting a notion of stochastic dominance.
Ultimately, some pairs of players can be conclusively compared while others cannot.
This approach departs from existing methods in that it relies on ranking data (rather than match-level data, which may or may not exist/be available) and formally comes to terms with the possibility that some comparisons are ``too close to call.''
Our methods leave ample room for further modeling choices, namely by adapting specifics of the constructions \eqref{eq: P} and \eqref{eq: dominance} as needed for the particular application in mind.
In general, this works motivates the analytical study of data-driven sampling of permutations, be it with transpositions or other shuffling techniques.
We anticipate this to be both challenging and of independent mathematical interest.

We implemented our methods using data from both the ATP and the WTA to find the GOATs in the respective categories.
Our experimental results suggest that the ``Big Three'' of the ATP, that is Nadal, Federer, and Djokovic, are the ones that tend to come ahead as GOATs under most circumstances, whereas S. Williams and Graf are the ones that do so in the case of the WTA.
Through our experiments, we also noted major differences in terms of career and
dominance longevity for top players across the two
associations.

Finally we remark that our methods, while initially motivated by this application in sports analytics, can be applied more broadly to other areas of operations management. 
For example, consider the area of route planning logistics, wherein many problems involve the frequent re-optimization of routes that start and end at fixed depots. 
In these settings, historical routes (e.g., optimized on a daily basis) can be compiled into a collection, in much in the same way as the historical rankings in Figure~\ref{fig: historical}\textemdash after all, a route is a permutation over locations to be visited.
Therefore, methods similar to ours could produce a poset over locations, in this way ``learning'' (perhaps soft) precedence constraints from historical data.
The distance between locations in terms of edges in Hasse diagram could intuitively encode the level of confidence in these constraints.
Such information could be useful for the purpose of speeding-up integer programming solvers, through the addition of cuts that discard routes of poor quality.
There are other examples of the usefulness of inferring precedence constraints, for example those stemming from the 2021 Amazon Last Mile Routing Research Challenge~\cite{merchan20242021}.
Cook, Held, and Helsgaun~\cite{cook2024constrained} came first in the challenge by using historical last-mile delivery data to infer precedence constraints at a ``zone'' level (i.e., a higher level of aggregation than individual locations), in this way producing routes that closely-match those deemed of ``high quality'' by seasoned drivers who account for various practical aspects beyond merely route length.
We therefore believe an extension of the ideas presented in this work could fall in line with recent calls~~\cite{winkenbach2024introduction} for ``innovative solution methods'' in route planning.

\section*{Acknowledgements}
J. C. Mart\'inez Mori would like to thank Patrick Kastner for computational resources and Amanda Priestley for a helpful discussion.
Part of this research was performed while J. C. Mart\'inez Mori was visiting the Mathematical Sciences Research Institute (MSRI), now becoming the Simons Laufer Mathematical Sciences Institute (SLMath), which is supported by NSF Grant No. DMS-192893.
J. C. Mart\'inez Mori is supported by Schmidt
Science Fellows, in partnership with the Rhodes Trust.
J. C. Mart\'inez Mori is supported by the President's Postdoctoral Fellowship Program (PPFP) at the Georgia Institute of Technology.

\printbibliography

\end{document}